\documentclass[10pt]{article}
\usepackage[utf8]{inputenc}
\usepackage[T1]{fontenc}
\usepackage[english]{babel}
\usepackage{amsmath,amsthm,setspace, amsfonts, amssymb,graphics,multirow,enumerate}
\usepackage{todonotes}
\usepackage{graphicx,color}
\usepackage[linesnumbered,vlined,boxed,commentsnumbered]{algorithm2e}

 \usepackage{pdfpages}
\usepackage{float}
\usepackage{listings}
 \usepackage{bbm}
\usepackage{calc}
\usepackage{array}
 \usepackage{authblk}
 \usepackage{natbib}

\newtheorem{theorem}{Theorem}
\newtheorem{proposition}[theorem]{Proposition}
\newtheorem{lemma}[theorem]{Lemma}

\newtheorem{assumption}{H\hspace{-3pt}}

\newcommand{\rset}{\mathbb{R}}
\newcommand{\nset}{\mathbb{N}}
\newcommand{\obs}{\mathsf{Y}}
\newcommand{\hid}{Z}
\newcommand{\Xset}{\mathsf{Z}}
\newcommand{\Xsigma}{\mathcal{Z}}
\newcommand{\pscal}[2]{\left<#1,#2\right>}
\newcommand{\rmd}{\mathrm{d}}
\newcommand{\eqdef}{:=}

\newcommand{\mc}{\mathrm{mc}}
\newcommand{\sa}{\mathrm{sa}}
\newcommand{\barS}{\overline{S}}
\newcommand{\hatS}{\widehat{S}}
\newcommand{\Q}{\mathcal{Q}}

\newcommand{\calL}{\mathcal{L}}
\newcommand{\A}{\mathsf{A}}
\newcommand{\F}{\mathcal{F}}
\newcommand{\PE}{\mathbb{E}}
\newcommand{\PP}{\mathbb{P}}
\newcommand{\1}{\mathbbm{1}}
\newcommand{\Wnorm}[1]{\left| #1 \right|_W}
\newcommand{\sqrtWnorm}[1]{\left| #1 \right|_{\sqrt{W}}}
\newcommand{\WnormM}[1]{\left\| #1 \right\|_W}
\newcommand{\sqrtWnormM}[1]{\left\| #1 \right\|_{\sqrt{W}}}
\newcommand{\D}{\mathsf{D}}

\newcommand{\coint}[1]{\left[#1\right)}
\newcommand{\ocint}[1]{\left(#1\right]}
\newcommand{\ooint}[1]{\left(#1\right)}
\newcommand{\ccint}[1]{\left[#1\right]}
\newcommand{\jacob}[2]{\operatorname{J} #1(#2)}
\def\argmin{\operatorname{argmin}}
\def\argmax{\operatorname{argmax}}
\def\Prox{\operatorname{Prox}}

\begin{document}

\title{Stochastic Proximal Gradient Algorithms  \\ for Penalized Mixed Models}


\author[1]{Gersende Fort}
\author[2,3]{Edouard Ollier}
\author[4]{Adeline Samson}
\affil[1]{IMT UMR5219, Universit\'e de Toulouse, CNRS; F-31062 Toulouse Cedex 9, France.}
\affil[2]{INSERM, U1059, Dysfonction Vasculaire et H\'emostase, Saint Etienne, France.}
\affil[3]{U.M.P.A., Ecole Normale Sup\'erieure de Lyon, CNRS UMR 5669; INRIA, Project-team NUMED. 46 All\'ee d'Italie, 69364 Lyon Cedex 07, France}
\affil[4]{Universit\'e Grenoble-Alpes, Laboratoire Jean Kuntzmann, UMR CNRS 5224}

\maketitle

\begin{abstract}
  Motivated by penalized likelihood maximization in complex models, we
  study optimization problems where neither the function to optimize
  nor its gradient have an explicit expression, but its gradient can
  be approximated by a Monte Carlo technique.  We propose a new
  algorithm based on a stochastic approximation of the
  Proximal-Gradient (PG) algorithm. This new algorithm, named
  Stochastic Approximation PG (SAPG) is the combination of a
  stochastic gradient descent step which - roughly speaking - computes
  a smoothed approximation of the gradient along the iterations, and a
  proximal step.  The choice of the step size and of the Monte Carlo
  batch size for the stochastic gradient descent step in SAPG are
  discussed. Our convergence results cover the cases of biased and
  unbiased Monte Carlo approximations. While the convergence analysis
  of some classical Monte Carlo approximation of the gradient is
  already addressed in the literature
  \citep[see][]{atchade2014stochastic}, the convergence analysis of SAPG
  is new.  Practical implementation is discussed and guidelines to
  tune the algorithm are given. The two algorithms are compared on a
  linear mixed effect model as a toy example. A more challenging
  application is proposed on non-linear mixed effect models in high
  dimension with a pharmacokinetic data set including genomic
  covariates. To our best knowledge, our work provides the first
  convergence result of a numerical method designed to solve penalized
  Maximum Likelihood in a non-linear mixed effect model.

\end{abstract}

\section{Introduction}
Many problems in computational statistics reduce to the maximization of a criterion
\begin{equation}\label{eq:problem}
\argmax_{\theta \in \rset^d} F(\theta), \qquad \text{where} \ F \eqdef \ell -
g,
\end{equation}
and the functions $\ell, g$ satisfy
\begin{assumption}
  \label{A1} the function $g:\;\rset^d\to \ccint{0,+\infty}$ is
  convex, not identically $+\infty$, and lower semi-continuous.
\end{assumption}
\begin{assumption}
  \label{A1bis}  the function $\ell: \rset^d \to \rset \cup \{-\infty\}$ is continuously differentiable
  on $\Theta \eqdef\{\theta \in \rset^d: g(\theta) + |\ell(\theta)| < \infty
  \}$ and its gradient is of the form 
\begin{equation}\label{eq:H1bis:gdt}
\begin{split}
  \nabla \ell(\theta) &= \nabla \phi(\theta) + \Psi(\theta) \barS(\theta),\\ &\qquad \text{with} \ \barS(\theta) \eqdef \int_\Xset S(z)
  \pi_\theta(z) \nu(\rmd z);
  \end{split}
\end{equation}
$\nabla$ denotes the gradient operator and $\pi_\theta \rmd \nu$ is a
probability distribution on a measurable subset $(\Xset, \Xsigma)$ of $\,
\rset^p$. The measurable functions $\nabla \phi: \rset^d \to \rset^d$ and
$\Psi: \rset^d\to \rset^{d \times q}$ are known but the expectation $\barS$ of
the function $S: \Xset \to \rset^q$ with respect to $\pi_\theta \rmd \nu$ may
be intractable.  Furthermore, there exists a finite non-negative constant $L$
such that for all $\theta, \theta' \in \Theta$,
    \begin{equation}\label{eq:H1bis:gdtLip}
    \| \nabla \ell(\theta) - \nabla \ell(\theta')\| \leq L \|\theta - \theta
    '\|;
    \end{equation}
    $\| \cdot\|$ is the Euclidean norm.
\end{assumption}
Examples of functions $\ell$ satisfying Eq.~(\ref{eq:H1bis:gdt}) are given below.
We are interested in numerical methods for solving Eq.~\eqref{eq:problem}, robust to the case when neither $\ell$
nor its gradient have an explicit expression.

\bigskip

Such an optimization problem occurs for example when computing a
penalized maximum likelihood estimator in some parametric model
indexed by $\theta \in \rset^d$: $\ell$ denotes the log-likelihood of
the observations $\obs$ (the dependence upon $\obs$ is omitted) and
$g$ is the penalty term.

The optimization problem Eq.~\eqref{eq:problem} covers the computation of the
maximum when the parameter $\theta$ is restricted to a closed convex subset
$\Theta$ of $\rset^d$; in that case, $g$ is the characteristic function of  $\Theta$ \textit{i.e.}  $g(\theta) =0$ for any $\theta \in \Theta$ and $g(\theta) =+
\infty$ otherwise.  It also covers the case when $g$ is the ridge, the lasso or
the elastic net penalty; and more generally, the case when $g$ is the sum of
lower semi-continuous non-negative convex functions.

A first example of such a function $\ell$ is given by the
log-likelihood in a latent variable model with complete likelihood
from the $q$-parameter exponential family (see
\textit{e.g.}~\cite{bickel:doksum:2015} and \cite{bartholomew:2011} and the
references therein).  In that case, $\ell$ is of the form
\begin{equation} \label{eq:intro:latent}
\theta \mapsto \ell(\theta) \eqdef \log \int_\Xset \exp\left( \phi(\theta) +
  \pscal{S(z)}{\psi(\theta)} \right) \nu(\rmd z),
\end{equation}
where $\pscal{a}{b}$ denotes the scalar of two vectors $a,b \in \rset^l$,
$\phi: \rset^d \to \rset$, $\psi: \rset^d\to \rset^q$ and $S: \Xset \to
\rset^q$ are measurable functions, and $\nu$ is a $\sigma$-finite positive
measure on $(\Xset,\Xsigma)$.  The quantity $\theta \mapsto \phi(\theta) +
\pscal{S(\hid)}{\psi(\theta)}$ is known as the complete log-likelihood, and
$\hid$ is the latent data vector.  Under regularity conditions, we have
\begin{equation}\label{eq:intro:latent:gdt}
\begin{split}
\nabla  \ell(\theta) &=  \nabla \phi(\theta) + \jacob{\psi}{\theta}  \int_\Xset S(z) \ 
\pi_{\theta}(z) \nu(\rmd z), \\ & \text{with} \ \pi_{\theta}(z) \eqdef
\frac{\exp(\pscal{S(z)}{\psi(\theta)})}{\int_\Xset \exp(\pscal{S(u)}{\psi(\theta)}) \nu(\rmd u)},
\end{split}
\end{equation}
where $\jacob{\psi}{\theta}$ denotes the transpose of the jacobian matrix of the function $\psi$ at $\theta$.

A second example is given by the
log-likelihood of $N$ independent observations
$(\obs_{1}, \cdots, \obs_{N})$ from a log-linear model for Markov
random fields. 
In this model, $\ell$ is given by
\begin{equation}\label{eq:intro:gibbs}
\begin{split}
\theta \mapsto &\ell(\theta)\eqdef  \\ &\sum_{k=1}^N \pscal{S(\obs_{k})}{\theta} - N \log \int_\Xset \exp\left( \pscal{S(z)}{\theta} \right) \nu(\rmd z).
\end{split}
\end{equation}
The function $\theta \mapsto \int_\Xset \exp\left( \pscal{S(z)}{\theta} \right)
\nu(\rmd z)$ is known as the partition function.  Under regularity
conditions, we have
\begin{equation}\label{eq:intro:gibbs:gdt}
\begin{split}
\nabla \ell(\theta) &= \sum_{k=1}^N S(\obs_{k} ) - N \int_\Xset S(z) \ \pi_\theta(z) \nu(\rmd z),
\\ &\qquad \text{with} \, \pi_\theta(z) \eqdef \frac{\exp\left(
    \pscal{S(z)}{\theta} \right)}{\int \exp\left( \pscal{S(u)}{\theta} \right)
  \, \nu(\rmd u)}.
  \end{split}
\end{equation}

\bigskip 

In these two examples, the integrals in Eqs.~\eqref{eq:intro:latent}
to \eqref{eq:intro:gibbs:gdt} are intractable except for toy examples:
neither the function $\ell$ nor its gradient are available.
Nevertheless, all the integrals in Eqs.~
\eqref{eq:intro:latent}-\eqref{eq:intro:gibbs:gdt} can be approximated
by a Monte Carlo sum  \citep[see \textit{e.g.}][]{casella:robert:2004}.  In the
first example, this Monte Carlo approximation consists in imputing the
missing variables $z$; it is known that such an imputation is far more
efficient when the Monte Carlo samples are drawn under
$\pi_\theta \rmd \nu$, \textit{i.e.} the \textit{a posteriori} distribution of the
missing variables given the observations (see
Eq.~\eqref{eq:intro:latent:gdt}) than when they are drawn under the \textit{a
priori} distribution. This remark is the essence of the Expectation
Maximization (EM) algorithm \cite[introduced in][]{dempster:laird:rubin:1977}, a popular iterative procedure for
maximizing  the log-likelihood $\ell$ in latent variable models.

In this paper, we are interested in first order optimization methods to solve
Eq.~\eqref{eq:problem}, that is methods based on the gradient. In
Section~\ref{sec:algos}, we describe two stochastic first-order descent
methods, which are stochastic perturbations of the Proximal-Gradient (PG)
algorithm (introduced in ~\cite{combettes:pesquet:2011}; see also
~\cite{Beck2009,Parikh2013} for literature reviews on Proximal-Gradient
algorithms). The two algorithms are the Monte Carlo Proximal-Gradient algorithm
(MCPG) and the Stochastic Approximation Proximal-Gradient algorithm (SAPG),
which differ in the approximation of the gradient $\nabla \ell$ and more
precisely, of the intractable integral $\barS(\theta)$ (see
Eq.~\eqref{eq:H1bis:gdt}).  In MCPG, at each iteration $n$ of the algorithm,
this expectation evaluated at the current point $\theta_n$ is approximated by a
Monte Carlo sum computed from samples $\{\hid_{1,n}, \cdots,
\hid_{m_{n+1},n}\}$ approximating $\pi_{\theta_n} \rmd \nu$.  In SAPG, the
approximation is computed as a Monte Carlo sum based on all the points drawn
during all the previous iterations of the algorithm $\{\hid_{i,j}, i \leq m_{j+1},
j \leq n\}$.

When $\ell$ is the log-likelihood of a latent variable model, we prove in
Section~\ref{sec:algos:gdtproxEM} that our algorithms are Generalized EM
algorithms \cite[see \textit{e.g.}][]{McLachlan:Krishnan:2008,ng:krishnan:2012}
combined with a stochastic E-step: in MCPG and SAPG, the stochastic E-step
mimics respectively the E-step of the Monte Carlo EM
\citep{wei:tanner:1990,levine:fan:2004} and the E-step of the Stochastic
Approximation EM \cite[see \textit{e.g.}][]{delyon1999convergence}.

Section~\ref{sec:theorie} is devoted to the convergence analysis of MCPG and
SAPG.  These algorithms can be seen as perturbed Proximal-Gradient algorithms
when the perturbation comes from replacing the exact quantity $\barS(\theta_n)$
by a Monte Carlo approximation $S_{n+1}$ at each iteration of the algorithm.
Our convergence analysis covers the case when the points $\{\hid_{1,n}, \cdots,
\hid_{m_{n+1},n}\}$ are sampled from a Markov chain Monte Carlo sampler (MCMC)
with target distribution $\pi_{\theta_n} \rmd \nu$ - and therefore, it also
covers the case of i.i.d. draws. This implies that the estimator $S_{n+1}$ of
$\barS(\theta_n)$ may be biased. There exist many contributions in the
literature on the convergence of perturbed Proximal-Gradient algorithms when
$\ell$ is concave, but except in the works by~\cite{atchade2014stochastic} and
\cite{combettes:pesquet:2015:biased}, most of them assume that the error
$S_{n+1} - \barS(\theta_n)$ is unbiased and gets small when $n \to \infty$ \citep[see
\textit{e.g.}][]{Rosasco:Villa:Vu:2014,combettes:pesquet:2015,rosasco:villa:vu:2015,lin:rosasco:villa:zhou:2015}.
In this paper, we provide sufficient conditions for the almost-sure convergence
of MCPG and SAPG under the assumption that $\ell$ is concave and with no
assumptions on the bias of $S_{n+1} - \barS(\theta_n)$. The convergence
analysis of MCPG is a special case of \cite[Section 4]{atchade2014stochastic};
to our best knowledge, the convergence of SAPG is a new result.

Practical implementation is discussed in Section \ref{sec:practical}. Some
guidelines are given in
Section~\ref{sec:guidelines} to choose the sequences involved in the stochastic
approximation procedures.  Then, MCPG and SAPG are compared through a toy
example in Section~\ref{sec:simuLongTime}. A more challenging application to penalized
inference in a mixed effect model is detailed in
Section~\ref{sec:applications}. Mixed models are applied to analyze repeated
data in a population of subjects.  The $N$ independent vectors of observations $( \obs_k, k=1, \ldots,
N)$ of the $N$ subjects are modeled by
\begin{equation}\label{eq:ModeleMixte}
  \obs_k = f(t_k, \hid^{(k)}) + \varepsilon_k,
  \end{equation}
  with individual latent variable $\hid^{(k)}$ independent of the
  measurement error vector $\varepsilon_k$ and $f$ the regression
  function that depends on the vector of observation times $t_k$.
  Mixed models thus enter the class of models given by Eq.~(\ref{eq:intro:latent})
  with latent variables $\hid = (\hid^{(1)}, \ldots, \hid^{(N)})$.
  When a covariate model is introduced, the number of covariates can be large,
  but with only a few of them being influential. This is a sparse estimation
  problem and the selection problem can be treated through the optimization of
  a penalized version of the log-likelihood Eq.~(\ref{eq:intro:latent}).  In
  non-linear mixed models, the optimization problem is not explicit and
   stochastic penalized versions of EM \citep{bertrand2013multiple,
    ollier2016saem, chen2017NLMEM} have been proposed. To our best knowledge,
  stochastic Proximal-Gradient algorithms have not been proposed for mixed
  models.

\section{Stochastic Proximal-Gradient based algorithms}
In this section, we describe first-order based algorithms for solving
Eq.~\eqref{eq:problem} under the assumptions H\ref{A1} and H\ref{A1bis},
when the expectation $\barS(\theta)$ in Eq.~\eqref{eq:H1bis:gdt} is intractable.

\subsection{The MCPG and SAPG algorithms}
\label{sec:algos}
Both MCPG and SAPG are iterative algorithms, each update relies on the
combination of a gradient step and a proximal operator.  The proximal map
(\cite{Moreau1962}, see also \cite{bauschke:combettes:2011,Parikh2013})
associated to a convex function $g$ is defined for any $\gamma>0$ and $\theta
\in \rset^d$ by
\begin{equation}\label{eq:def:Proximal}
\Prox_{\gamma,g}(\theta)  \eqdef \argmin_{\tau \in \Theta} \left\{ g(\tau) + \frac{1}{2\gamma}
  \| \theta - \tau \|^2 \right\}.
\end{equation}
Note that under H\ref{A1}, for any $\gamma >0$ and $\theta \in \rset^d$, there
exists an unique point $\tau$ minimizing the RHS of
Eq.~\eqref{eq:def:Proximal}. This proximal operator may have an explicit
expression. When $g$ is the characteristic function
\[
g(\theta) \eqdef \left\{
  \begin{array}{ll}
    0 &   \text{if $ \theta \in \Theta$} \\
    + \infty & \text{otherwise},
  \end{array}
\right.
\]
for some closed convex set $\Theta \subseteq \rset^d$, then $g(\theta)$ is the
projection of $\theta$ on $\Theta$. This projection is explicit for example
when $\Theta$ is an hyper-rectangle. Another example of explicit proximal
operator is the case associated to the so-called elastic net penalty
\textit{i.e.}  $g_{\lambda,\alpha}(\theta) \eqdef \lambda \left(
  \frac{1-\alpha}{2} \sum_{i=1}^d \theta_i^2\right.$ $ \left.+ \alpha
  \sum_{i=1}^d |\theta_i| \right)$ with $\theta=(\theta_1, \cdots, \theta_d)$,
$\lambda >0$ and $\alpha \in \ocint{0,1}$, then for any component $i \in \{1,
\cdots, d \}$,
\begin{multline*}
  \left( \Prox_{\gamma,g_{\lambda, \alpha}}(\theta) \right)_i  \\
  = \frac{1}{1+\gamma \lambda (1-\alpha)} \left\{
  \begin{array}{ll}
    0 & \text{if $ |\theta_i| \leq \gamma \lambda \alpha$,} \\
    \theta_i - \gamma \lambda \alpha & \text{if $\theta_i \geq \gamma \lambda \alpha$,} \\
\theta_i + \gamma \lambda \alpha & \text{if $\theta_i \leq -\gamma \lambda \alpha$.}
  \end{array}
\right.
\end{multline*}

\bigskip

The Proximal-Gradient algorithm for solving the optimization problem
Eq.~\eqref{eq:problem} produces a sequence $\{\theta_n, n \geq 0 \}$ as
follows: given a $\ocint{0,1/L}$-valued sequence $\{\gamma_n, n \geq 0 \}$,
\begin{equation}\label{eq:PG:iteration}
\begin{split}
\theta_{n+1} &= \Prox_{\gamma_{n+1},g} \left( \theta_n + \gamma_{n+1}
  \nabla \ell(\theta_n) \right) \\ &= \Prox_{\gamma_{n+1},g} \left(
  \theta_n + \gamma_{n+1} \{ \nabla \phi(\theta_n) + \Psi(\theta_n)
  \barS(\theta_n) \} \right).
 \end{split}
\end{equation}
This update scheme can be explained as follows: by H\ref{A1bis}, we have for
any $L \leq \gamma_{n+1}^{-1}$,
\[
\begin{split}
&F(\theta) = \ell(\theta) -g(\theta) \\ &\geq \ell(\theta_n) -
\pscal{\nabla \ell(\theta_n)}{\theta-\theta_n} - \frac{1}{2
  \gamma_{n+1}} \|\theta - \theta_n \|^2 -g(\theta).
\end{split}
\]
This minorizing function is equal to $F(\theta_n)$ at the point $\theta_n$; the
maximization (w.r.t. $\theta$) of the RHS yields $\theta_{n+1}$ given by Eq.~
\eqref{eq:PG:iteration}. The Proximal-Gradient algorithm is therefore a
Minorize - Majorization (MM) algorithm and the ascent property holds:
$F(\theta_{n+1}) \geq F(\theta_n)$ for all $n$. Sufficient conditions for the
convergence of the Proximal-Gradient algorithm Eq.~\eqref{eq:PG:iteration} can be
derived from the results by \cite{combettes:wajs:2005,Parikh2013} or from
convergence analysis of MM algorithms 
\cite[see \textit{e.g.}][]{zangwill:1969,meyer:1976}.

\bigskip

In the case $\barS(\theta)$ can not be computed, we describe two strategies for
a Monte Carlo approximation. At iteration $n+1$, given the current value of the
parameter $\theta_n$, $m_{n+1}$ points $\{\hid_{1,n}, \cdots, \hid_{m_{n+1},n}
\}$ from the path of a Markov chain with target distribution $\pi_{\theta_n}
\rmd \nu$ are sampled. A first strategy consists in replacing $\barS(\theta_n)$
by a Monte Carlo mean:
\begin{equation}
  \label{eq:MonteCarloApprox}
  S_{n+1}^{\mc} \eqdef \frac{1}{m_{n+1}} \sum_{j=1}^{m_{n+1}} S(\hid_{j,n}).
\end{equation}
A second strategy, inspired by stochastic approximation methods \cite[see \textit{e.g.}][]{benveniste:metivier:priouret:1990,kushner:yin:2003} consists in
replacing $\barS(\theta_n)$ by a stochastic approximation
\begin{equation}
  \label{eq:SAApprox}
  S_{n+1}^{\sa} \eqdef  (1-\delta_{n+1}) S_n^{\sa} + \frac{\delta_{n+1}}{m_{n+1}} \sum_{j=1}^{m_{n+1}} S(\hid_{j,n}),
\end{equation}
where $\{\delta_n, n \geq 0 \}$ is a deterministic
$\ccint{0,1}$-valued sequence. These two strategies yield respectively
the Monte Carlo Proximal-Gradient (MCPG) algorithm (see
Algorithm~\ref{algo:MCPG}) and the Stochastic Approximation
Proximal-
Gradient (SAPG) algorithm (see Algorithm~\ref{algo:SAPG}).

\noindent\begin{minipage}{\linewidth}
\begin{algorithm}[H]
  \SetKwInOut{Input}{input} \SetKwInOut{Output}{output} \Input{The
    initial values $\theta_0 \in \Theta$ and  $\hid_{m_0,-1} \eqdef z_\star$, a $\ocint{0,1/L}$-valued
    sequence $\{\gamma_n, n \geq 0\}$ and an integer valued sequence $\{m_n, n\geq 0\}$} \Output{The sequence
    $\{\theta_n, n\geq 0 \}$} \For{$n \geq 0$}{ \emph{Simulation-step}
    \; sample a path $\hid_{1,n}, \cdots, \hid_{m_{n+1},n}$ of a Markov chain
    with invariant distribution $\pi_{\theta_n} \rmd \nu$ and started from $\hid_{m_{n},n-1}$\;
    \emph{Expectation step} \; Compute $S_{n+1}^\mc$ as in Eq.~\eqref{eq:MonteCarloApprox}
    \label{line:algo:EM:Estep} \; \emph{Maximization step} \;
    Set $\theta_{n+1} = \Prox_{\gamma_{n+1},g}\left( \theta_n +
      \gamma_{n+1} \{\nabla \phi(\theta_n) + \Psi(\theta_n) S_{n+1}^\mc
      \}\right)$
\label{line:algo:EM:Mstep} }
\caption{The Monte Carlo Proximal-Gradient algorithm for the
  maximization of $\ell -g$ \label{algo:MCPG}}
\end{algorithm}
\end{minipage}

\noindent\begin{minipage}{\linewidth}
\begin{algorithm}[H]
  \SetKwInOut{Input}{input} \SetKwInOut{Output}{output} \Input{The initial
    values $\theta_0 \in \Theta$,  $\hid_{m_0,-1} \eqdef z_\star$ and $S_0^\sa \eqdef s_\star$, a $\ocint{0,1/L}$-valued sequence $\{\gamma_n,
    n \geq 0\}$, a $\ccint{0,1}$-valued sequence $\{\delta_n, n\geq 0\}$ and an
    integer valued sequence $\{m_n, n\geq 0\}$} \Output{The sequence
    $\{\theta_n, n\geq 0 \}$} \For{$n \geq 0$}{ \emph{Simulation-step} \;
    sample a path $\hid_{1,n}, \cdots, \hid_{m_{n+1},n}$ of a Markov chain with
    invariant distribution $\pi_{\theta_n} \rmd \nu$ and started from $\hid_{m_{n},n-1}$\; \emph{Expectation step}
    \; Compute $S_{n+1}^\sa$ as in Eq.~\eqref{eq:SAApprox}
    \label{line:algo:SAPG:Estep} \; \emph{Maximization step} \;
    Set $\theta_{n+1} = \Prox_{\gamma_{n+1},g}\left( \theta_n + \gamma_{n+1}
      \{\nabla \phi(\theta_n) + \Psi(\theta_n) S_{n+1}^\sa \}\right)$
\label{line:algo:SAPG:Mstep} }
\caption{The Stochastic Approximation Proximal-Gradient algorithm for
  the maximization of $\ell -g$ \label{algo:SAPG}}
\end{algorithm}
\end{minipage}

In Section~\ref{sec:theorie}, we prove the convergence of MCPG to the
maximum points of $F$ when $\ell$ is concave, for different choices of
the sequences $\{\gamma_n, m_n, n\geq 0\}$ including decreasing or
constant step sizes $\{\gamma_n,n \geq 0 \}$ and respectively, constant
or increasing batch size $\{m_n, n \geq 0\}$.  We also establish the
convergence of SAPG to the maximum points (in the concave case); only
the case of a constant batch size $\{m_n,n \geq 0\}$ and a decreasing
step size $\{\gamma_n, n \geq 0\}$ is studied, since this framework
corresponds to the {\em Stochastic Approximation} one from which the
update rule Eq.~\eqref{eq:SAApprox} is inherited \cite[see details
in][]{delyon1999convergence}.  From a numerical point of view, the
choice of the sequences $\{\gamma_n, n\geq 0\}$,
$\{\delta_n, n \geq 0 \}$ and $\{m_n, n \geq 0\}$ is discussed in
Section~\ref{sec:practical}: guidelines are given in
Section~\ref{sec:guidelines} and the behavior of the algorithm is
illustrated through a toy example in
Section~\ref{sec:simuLongTime}.

\subsection{Case of latent variable models from the exponential family}
\label{sec:algos:gdtproxEM}
In this section, we consider the case when $\ell$ is given by
Eq.~\eqref{eq:intro:latent}. A classical approach to solve penalized maximum
likelihood problems in latent variables models with complete likelihood from
the exponential family is the Expectation-Maximization (EM) algorithm or a
generalization called the Generalized EM (GEM) algorithm
\citep{dempster:laird:rubin:1977,McLachlan:Krishnan:2008,ng:krishnan:2012}.  Our
goal here, is to show that MCPG and SAPG are stochastic perturbations of a GEM
algorithm.

\bigskip

The EM algorithm is an iterative algorithm: at each iteration, given
the current parameter $\theta_n$, the quantity
$\Q(\theta \vert \theta_n)$, defined as the conditional expectation of
the complete log-likelihood under the \textit{a posteriori} distribution for
the current fit of the parameters, is computed:
\begin{equation}\label{eq:def:QdeEM}
  \Q(\theta \vert \theta') \eqdef \phi(\theta) + \pscal{\barS(\theta')}{\psi(\theta)}.
\end{equation}
The EM sequence $\{\theta_n, n\geq 0\}$ for the maximization of the penalized
log-likelihood $\ell -g$ is given by   \citep[see][Section
1.6.1.]{McLachlan:Krishnan:2008}
\begin{equation} \label{eq:EM:iteration}
\theta_{n+1} = \argmax_{\theta \in \Theta} \left\{ \phi(\theta) +
  \pscal{\barS(\theta_n)}{\psi(\theta)} -g(\theta) \right\}.
\end{equation}
When $\barS(\theta)$ is intractable, it was proposed to replace
$\barS(\theta_n)$ in this EM-penalized algorithm by an approximation $S_{n+1}$
- see Algorithm~\ref{algo:EM}. When $S_{n+1} = S_{n+1}^{mc}$ (see
Eq.~\eqref{eq:MonteCarloApprox}), this yields the so-called Monte Carlo-EM
penalized algorithm (MCEM-pen), trivially adapted from MCEM proposed
by~\cite{wei:tanner:1990,levine:fan:2004}.  Another popular strategy is to
replace $\barS(\theta_n)$ by $S_{n+1}^\sa$ (see Eq.~\eqref{eq:SAApprox})
yielding to the so-called Stochastic Approximation-EM penalized algorithm
(SAEM-pen) - (see~\cite{delyon1999convergence} for the unpenalized version).

\noindent\begin{minipage}{\linewidth}
\begin{algorithm}[H]
  \SetKwInOut{Input}{input} \SetKwInOut{Output}{output} \Input{The
    initial value $\theta_0 \in \Theta$} \Output{The sequence
    $\{\theta_n, n\geq 0 \}$} \For{$n \geq 0$}{\emph{E-step} \;
    Compute an approximation $S_{n+1}$ of $\barS(\theta_n)$
    \label{line:algo:EM:Estep} \; \emph{M-step} \; Set $\theta_{n+1} = \argmax_{\theta} \{ \phi(\theta) + \pscal{S_{n+1}}{\psi(\theta)} -g(\theta) \}$ 
\label{line:algo:EM:Mstep} }
\caption{Perturbed EM-penalized algorithms for the maximization of $\ell -g$ \label{algo:EM}}
\end{algorithm}
\end{minipage}

When the maximization of Eq.~\eqref{eq:def:QdeEM} is not explicit, the update
of the parameter is modified as follows, yielding the Generalized EM-penalized
algorithm~(GEM-pen):

\begin{equation}\label{eq:GEM:iteration}
\begin{split}
  \theta_{n+1} \ \text{s.t.} \ &\phi(\theta_{n+1}) +
  \pscal{\barS(\theta_n)}{\psi(\theta_{n+1})} -g(\theta_{n+1}) \\ &\geq
  \phi(\theta_n) + \pscal{\barS(\theta_n)}{\psi(\theta_n)} -g(\theta_n).
  \end{split}
\end{equation}
This update rule still produces a sequence $\{\theta_n, n\geq 0\}$ satisfying
the ascent property $F(\theta_{n+1}) \geq F(\theta_n)$ which is the key
property for the convergence of EM  \citep[see \textit{e.g.}][]{wu:1983}.  Here
again, the approximations defined in Eq.~(\ref{eq:MonteCarloApprox}) and
Eq.~(\ref{eq:SAApprox}) can be plugged in the GEM-pen update
Eq.~\eqref{eq:GEM:iteration} when $\barS$ is not explicit.

We show in the following proposition that the sequence $\{\theta_n, n \geq 0\}$
produced by the Proximal-Gradient algorithm Eq.~\eqref{eq:PG:iteration} is a
GEM-pen sequence since it satisfies the inequality Eq.~\eqref{eq:GEM:iteration}. As
a consequence, MCPG and SAPG are stochastic GEM-pen algorithms.
\begin{proposition}
  \label{prop:monotonicity:EMprox}
  Let $g$ satisfying H\ref{A1} and $\ell$ be of the form
  Eq.~\eqref{eq:intro:latent} with continuously differentiable functions
  $\phi: \rset^d \to \rset$, $\psi:\rset^d \to \rset^q$ and
  $S: \Xset \to \rset^q$.  Set
  $\Theta \eqdef \{ g +|\ell| < \infty \}$. Define
  $\barS: \Theta \to \rset^q$ by
  $\barS(\theta) \eqdef \int_\Xset S(z) \, \pi_\theta(z) \, \nu(\rmd
  z)$
  where $\pi_\theta$ is given by Eq.~\eqref{eq:intro:latent:gdt}. Assume
  that there exists a constant $L>0$ such that for any
  $s \in \barS(\Theta)$, and any $\theta, \theta' \in \Theta$,
\begin{multline} \label{eq:HYP}
 \| \nabla \phi(\theta) - \nabla \phi(\theta') + \left(
  \jacob{\psi}{\theta} - \jacob{\psi}{\theta'}\right) s \|  \\ \leq L \|
\theta - \theta' \|.
\end{multline}
Let $\{\gamma_n, n \geq 0 \}$ be a (deterministic) positive sequence
such that $\gamma_n \in \ocint{0, 1/L}$ for all $n \geq 0$.  

Then the Proximal-Gradient algorithm Eq.~\eqref{eq:PG:iteration} is a GEM-pen
algorithm for the maximization of $\ell-g$.
\end{proposition}
The proof is postponed in Appendix~\ref{sec:proof:EMprox}. The assumption
Eq.~\eqref{eq:HYP} holds when $\Theta$ is compact and $\barS$ (resp.  $\phi$
and $\psi$) are continuous (resp. twice continuously differentiable). Note also
that for any $\theta$, $\theta^{'} \in \Theta$ and $s \in \barS(\Theta)$, we
have $\left(\jacob{\psi}{\theta} - \jacob{\psi}{\theta'}\right) s =0$ if
$\barS(\theta) \in \mathrm{Ker}( \jacob{\psi}{\theta'})$ for any $\theta$,
$\theta^{'} \in \Theta$.

\section{Convergence of MCPG and SAPG}
\label{sec:theorie}
The convergence of MCPG and SAPG is established by applying recent results
from~\cite{atchade2014stochastic} on the convergence of perturbed
Proximal-Gradient algorithms.  \cite[Theorem 2]{atchade2014stochastic} applied
to the case $\nabla \ell(\theta)$ is of the form $\nabla \phi(\theta) +
\Psi(\theta) \barS(\theta)$, where $\barS(\theta)$ is an intractable
expectation and $\nabla \phi, \Psi$ are explicit, yields
\begin{theorem} \label{theo:Cvg:perturbedEMprox} Assume H\ref{A1},
  H\ref{A1bis}, $\theta \mapsto \ell(\theta)$ is concave, and the set $\calL
  \eqdef \argmax_{\theta \in \Theta} F(\theta)$ is a non empty subset of
  $\Theta$.  Let $\{\theta_n, n \geq 0 \}$ be given by
  \[
\theta_{n+1} = \Prox_{\gamma_{n+1},g}\left( \theta_n + \gamma_{n+1}\left\{\nabla \phi(\theta_n) + \Psi(\theta_n) S_{n+1} \right\} \right),
\]
with a $\ocint{0,1/L}$-valued stepsize sequence
$\{\gamma_n, n \geq 0 \}$ satisfying $ \sum_n \gamma_n = + \infty$. If
the  series 
  \begin{align*}
    & \sum_n \gamma_{n+1}   \pscal{ \Psi(\theta_n) \left(S_{n+1} - \barS(\theta_n) \right)}{ T_{\gamma_{n+1}}(\theta_n) }, \\
    & \sum_n \gamma_{n+1}  \Psi(\theta_n) \left(S_{n+1} - \barS(\theta_n) \right), \\
    & \sum_n \gamma_{n+1}^2 \| \Psi(\theta_n) \left(S_{n+1} - \barS(\theta_n) \right) \|^2,
  \end{align*}
  converge,  where
\[
T_\gamma(\theta) \eqdef \Prox_{\gamma,g}(\theta +\gamma\left\{\nabla
  \phi(\theta) + \Psi(\theta) \barS(\theta)\right\} ),
\]
then there exists $\theta_\infty \in \calL$ such that $\lim_n
  \theta_n = \theta_\infty$.
\end{theorem}
We check the conditions of Theorem~\ref{theo:Cvg:perturbedEMprox} in
the case $S_{n+1}$ is resp. given by Eq.~\eqref{eq:MonteCarloApprox} for
the proof of MCPG and by Eq.~\eqref{eq:SAApprox} for the proof of
SAPG. Our convergence analysis is restricted to the case $\ell$ is
concave; to our best knowledge, the convergence of the perturbed
Proximal-Gradient algorithms when $\ell$ is not concave is an open
question. 

The novelty in this section is Proposition~\ref{prop:L2norm:DeltaS} and
Theorem~\ref{theo:Cvg:SAEMprox} which provide resp. a control of the $L_2$-norm
of the error $S_{n+1}^\sa - \barS(\theta_n)$ and the convergence of SAPG. These
results rely on a rewriting of $\left( S_{n+1}^\sa - \barS(\theta_n) \right)$
taking into account that $S_{n+1}^\sa$ is a weighted sum of the function $S$
evaluated at all the samples $\{\hid_{i,j}, i \leq m_{j+1}, j \leq n \}$ drawn
from the initialization of the algorithm. This approximation differs from a
more classical Monte Carlo approximation (see
Theorems~\ref{theo:MCPG:fixedbatch} and \ref{theo:MCPG:increasingbatch} for the
convergence of MCPG, which are special cases of the results
in~\cite{atchade2014stochastic}).

\bigskip

We allow the simulation step of MCPG and SAPG to rely on a Markov chain
Monte Carlo sampling: at iteration $(n+1)$, the conditional
distribution of $\hid_{j+1,n}$ given the past is
$P_{\theta_n}(\hid_{j,n}, \cdot)$ where $P_\theta$ is a Markov
transition kernel having $\pi_\theta \rmd \nu$ as its unique invariant
distribution.  The control of the quantities
$S_{n+1} - \barS(\theta_n)$ requires some ergodic properties on the
kernels $\{P_{\theta_n}, n \geq 0\}$ along the path
$\{\theta_n, n \geq 0 \}$ produced by the algorithm. These properties
have to be uniform in $\theta$, a property often called  the ``containment condition'' (see
\textit{e.g.} the literature on the convergence of adaptive MCMC samplers, for
example
\cite{andrieu:moulines:2006,roberts:rosenthal:2007,fort:moulines:priouret:2012}).
There are therefore three main strategies to prove the containment
condition.  In the first strategy, $\Theta$ is assumed to be bounded,
and a uniform ergodic assumption on the kernels
$\{P_\theta, \theta \in \Theta \}$ is assumed.  In the second one,
there is no boundedness assumption on $\Theta$ but the property
$\PP(\limsup_n \| \theta_n \| < \infty)=1$ has to be established \textit{prior}
the proof of convergence; a kind of local boundedness condition on the
sequence $\{\theta_n, n\geq 0\}$ is then applied - see \textit{e.g.}
\cite{andrieu:moulines:2006,fort:moulines:priouret:2012}.  The last
strategy consists in showing that
$\PP( \sup_n \rho_n \|\theta_n \| < \infty) =1$ for some deterministic
sequence $\{\rho_n, n \geq 0\}$ vanishing to zero when $n \to \infty$
at a rate compatible with the decaying ergodicity rate - see \textit{e.g.}
\cite{saksman:vihola:2010}.  The last two strategies are really
technical and require from the reader a strong background on
controlled Markov chain theory; for pedagogical purposes, we therefore
decided to state our results in the first context: we will assume that
$\Theta$ is bounded.

By allowing MCMC approximations, we propose a theory which covers the
case of a biased approximation, called below the {\em biased case}:
conditionally to the past
\begin{equation}\label{eq:filtration}
\F_n \eqdef \sigma\left(  \hid_{i,j}, i \leq m_{j+1}, j \leq n-1\right),
\end{equation}
the expectation of $S_{n+1}$ is not $\barS(\theta_n)$:
$\PE\left[ S_{n+1} \vert \F_n \right] \neq \barS(\theta_n)$.  As soon
as the samplers $\{P_\theta, \theta \in \Theta \}$ are ergodic enough
(for example, under H\ref{A4}\ref{A4:item1}) and
H\ref{A4}\ref{A4:item2})), the bias vanishes when the number of Monte
Carlo points $m_n$ tends to infinity.  Therefore, the proof for the
biased case when the sequence $\{m_n,n \geq 0 \}$ is constant is the
most technical situation since the bias does not decay.  It relies on
a specific decomposition of the error $S_{n+1} - \barS(\theta_n)$ into
a martingale increment with bounded $L^2$-moments, and a remainder
term which vanishes when $n \to \infty$ even when the batch size $m_n$
is constant. Such a behavior of the remainder term is a consequence of
regularity properties on the functions $\nabla \phi$, $\Psi$, $\barS$
(see H\ref{A3}\ref{A3:item3})), on the proximity operator (see
H\ref{A3}\ref{A3:item4})) and on the kernels
$\{P_\theta, \theta \in \Theta \}$ (see H\ref{A4}\ref{A4:item3})).

Our theory also
covers the {\em unbiased case} \textit{i.e.} when
$$\PE\left[ S_{n+1} \vert \F_n \right] = \barS(\theta_n).$$

\bigskip

We therefore establish the convergence of MCPG and SAPG by
strengthening the conditions H\ref{A1} and H\ref{A1bis} with
\begin{assumption}
  \label{A3}
  \begin{enumerate}[a)]
  \item \label{A3:item1} $\ell$ is concave and the set $\calL
    \eqdef \argmax_{\Theta} F$ is a non-empty subset of $\Theta$.
\item \label{A3:item2} $\Theta$ is bounded.
\item \label{A3:item3} There exists a constant $L$ such that for any
  $\theta,\theta' \in \Theta$,
\begin{multline*}
  \| \nabla \phi(\theta) - \nabla \phi(\theta') \| + \| \Psi(\theta) -
  \Psi(\theta') \| + \| \barS(\theta) - \barS(\theta')\|  \\
  \leq L \| \theta - \theta' \|,
\end{multline*}
where for a matrix $A$, $\|A\|$ denotes the operator norm associated
with the Euclidean vector norm.
\item \label{A3:item4} $\sup_{\gamma \in \ocint{0,1/L}} \sup_{\theta \in \Theta} \gamma^{-1} \| \Prox_{\gamma,g}(\theta) - \theta \| < \infty$.
  \end{enumerate}
\end{assumption}
Note that the assumptions
H\ref{A3}\ref{A3:item2})-H\ref{A3}\ref{A3:item3}) imply
Eq.~\eqref{eq:H1bis:gdtLip} and
$\sup_{\theta \in \Theta} \left( \| \nabla \phi(\theta) \| + \|
  \Psi(\theta) \| + \| \barS(\theta) \| \right) < \infty$.
When $\Theta$ is a compact convex set, then H\ref{A3}\ref{A3:item4})
holds for the elastic net penalty, the Lasso or the fused Lasso
penalty.  \cite[Proposition 11]{atchade2014stochastic} gives general
conditions for H\ref{A3}\ref{A3:item4}) to hold.

 \bigskip

 Before stating the ergodicity conditions on the kernels
 $\{P_\theta, \theta \in \Theta \}$, let us recall some basic
 properties on Markov kernels. A Markov kernel $P$ on the measurable
 set $(\Xset, \Xsigma)$ is an application on $\Xset \times \Xsigma$,
 taking values in $\ccint{0,1}$ such that for any $x \in \Xset$,
 $P(x,\cdot)$ is a probability measure on $\Xsigma$; and for any
 $A \in \Xsigma$, $x \mapsto P(x,A)$ is measurable.  Furthermore, if
 $P$ is a Markov kernel, $P^k$ denotes the $k$-th iterate of $P$
 defined by induction as 
 \begin{align*}
   & P^0(x,A) \eqdef \1_A(x), \\
& P^k(x,A) \eqdef \int P^{k-1}(x,\rmd z)P(z,A), \qquad k \geq 1.
 \end{align*}
 Finally, the kernel $P$ acts on the probability measures: for any
 probability measure $\xi$ on $\Xsigma$, $\xi P$ is a probability
 measure defined by
  \[
\xi P(A) \eqdef \int \xi(\rmd z) P(z,A), \qquad A \in \Xsigma;
\]
and $P$ acts on the positive measurable functions: for a measurable
function $f:\Xset \to \rset_+$, $Pf$ is a measurable function defined by
\[
Pf(z) \eqdef \int f(y) \, P(z, \rmd y).
\]
We refer the reader to \cite{meyn:tweedie:2009} for the definitions
and basic properties on Markov chains. Given a measurable function
$W: \Xset \to \coint{1,+\infty}$, define the $W$-norm of a signed measure $\nu$ on $\Xsigma$
and the $W$-norm of a function $f: \Xset \to \rset^d$:
\[
\Wnorm{f} \eqdef  \sup_{\Xset} \frac{\|f\|}{W}, \qquad \WnormM{\nu} \eqdef \sup_{f: \Wnorm{f} \leq 1} \left| \int f \rmd \nu \right|;
\]
these norms generalize resp. the supremum norm of a function and the
total variation norm of a measure.

Our results are derived under the following conditions on the kernels:
\begin{assumption}\label{A4}
\begin{enumerate}[a)]
\item \label{A4:item1} There exist $\lambda \in \ocint{0,1}$, $b < \infty$ and a measurable function $W: \Xset \to \coint{1,+\infty}$ such that
\[
\sqrtWnorm{S} < \infty, \qquad \sup_{\theta \in \Theta} P_\theta W \leq \lambda W + b.
\]
\item \label{A4:item2} There exist constants $C< \infty$ and
  $\rho \in \ooint{0,1}$ such that for any $z \in \Xset$ and
  $n \geq 0$,
\[
\sup_{\theta \in \Theta} \WnormM{ P_\theta^n(z, \cdot) - \pi_\theta } \leq C \, \rho^n \, W(z).
\]
\item \label{A4:item3} There exists a constant $C$ such that for any $\theta, \theta' \in \Theta$,
\begin{multline*}
\sqrtWnormM{\pi_\theta - \pi_{\theta'}} + \sup_{z \in \Xset} \frac{\sqrtWnormM{ P_\theta(z,\cdot) - P_{\theta'}(z,\cdot)}}{\sqrt{W}(z)}  \\
\leq C \, \|\theta - \theta' \|.
\end{multline*}
\end{enumerate}
\end{assumption}

Sufficient conditions for the uniform-in-$\theta$ ergodic behavior
H\ref{A4}\ref{A4:item2}) are given \textit{e.g.} in ~\cite[Lemma
2.3.]{fort:moulines:priouret:2011}: this lemma shows how to deduce
such a control from a minorization condition and a drift inequality on
the Markov kernels.  Examples of MCMC kernels $P_\theta$ satisfying
these assumptions can be found in \cite[Proposition
12]{andrieu:moulines:2006} and \cite[Proposition
15]{saksman:vihola:2010} for the adaptive Hastings-Metropolis
algorithm, in \cite[Proposition 3.1.]{fort:moulines:priouret:2011} for
an interactive tempering sampler, in \cite[Proposition
3.2.]{schreck:fort:moulines:2013} for the equi-energy sampler, and in
\cite[Proposition 3.1.]{fort:jourdain:kuhn:lelievre:stoltz:2015} for a
Wang-Landau type sampler.

\bigskip

Theorem~\ref{theo:MCPG:fixedbatch} establishes the convergence of MCPG
when the number of points in the Monte Carlo sum $S_{n+1}^\mc$ is
constant over iterations and the step size sequence
$\{\gamma_n, n \geq 0\}$ vanishes at a convenient rate. It is proved
in \cite[Theorem 4]{atchade2014stochastic}.
\begin{theorem}
  \label{theo:MCPG:fixedbatch}
  Assume H\ref{A1}, H\ref{A1bis}, H\ref{A3}\ref{A3:item1}-\ref{A3:item3}) and
  H\ref{A4}\ref{A4:item1}-\ref{A4:item2}). Let $\{\theta_n, n \geq 0 \}$ be the
  sequence given by Algorithm~\ref{algo:MCPG} with a $\ocint{0,1/L}$-valued
  sequence $\{\gamma_n, n \geq 0 \}$ such that $\sum_n \gamma_n = + \infty$ and
  $\sum_n \gamma_n^2<\infty$, and with a constant sequence $\{m_n, n \geq 0
  \}$.
  
  In the biased case, assume also H\ref{A3}\ref{A3:item4}) and
  H\ref{A4}\ref{A4:item3}) and $\sum_n |\gamma_{n+1} - \gamma_n |< \infty$.
  
  Then, with probability one, there exists $\theta_\infty \in \calL$ such that
  $\lim_n \theta_n = \theta_\infty$.
\end{theorem}

Theorem~\ref{theo:MCPG:increasingbatch} establishes the convergence of
MCPG when the number of points in the Monte Carlo sum $S_{n+1}^\mc$ is
increasing; it allows a constant stepsize sequence
$\{\gamma_n, n \geq 0\}$. It is proved in \cite[Theorem
6]{atchade2014stochastic}.
\begin{theorem}
  \label{theo:MCPG:increasingbatch}
  Assume H\ref{A1}, H\ref{A1bis}, H\ref{A3}\ref{A3:item1}-\ref{A3:item3}) and
  H\ref{A4}\ref{A4:item1}-\ref{A4:item2}). Let $\{\theta_n, n \geq 0 \}$ be the
  sequence given by Algorithm~\ref{algo:MCPG} with a $\ocint{0,1/L}$-valued
  sequence $\{\gamma_n, n \geq 0 \}$ and an integer valued sequence $\{m_n,n
  \geq 0\}$ such that $\sum_n \gamma_n = +\infty$ and $\sum_n
  \gamma_{n}^2/m_{n} < \infty$.
  
  In the biased case, assume also $\sum_n \gamma_{n}/m_{n}< \infty$.
  
  Then, with probability one, there exists $\theta_\infty \in \calL$ such that
  $\lim_n \theta_n = \theta_\infty$.
\end{theorem}
MCPG and SAPG differ in their approximation of $\barS(\theta_n)$ at each
iteration. We provide below a control of this error for a constant or a
polynomially increasing batch size $\{m_n, n \geq 0 \}$, and polynomially
decreasing stepsize sequences $\{\gamma_n, n \geq 0 \}$ and $\{\delta_n, n \geq
0 \}$.
\begin{proposition}
  \label{prop:L2norm:DeltaS} Let $\gamma_\star, \delta_\star, m_\star$
  be positive constants and $\beta \in \coint{0,1}$, $\alpha \geq \beta$,
  $c\geq0$. Set $\gamma_n = \gamma_\star n^{-\alpha}$,
  $\delta_n = \delta_\star n^{-\beta}$ and $m_n = m_\star
  n^{c}$. Assume H\ref{A1} to H\ref{A4}. Then
\[
\begin{split}
&\PE\left[ \| S_{n+1}^\mc - \barS(\theta_n)\|^2 \right] = O\left(
  n^{-c} \right), \\
 &\PE\left[ \| S_{n+1}^\sa - \barS(\theta_n)
  \|^2\right] =O\left(n^{- \{ 2(\alpha -\beta) \wedge (\beta+c)\}}
\right).
\end{split}
\]
\end{proposition}
The proof is given in Appendix~\ref{sec:proof:prop:Lnorm}.  This
proposition shows that when applying MCPG with a constant batch size
$(c=0)$, the error $S_{n+1}^\mc - \barS(\theta_n)$ does not vanish;
this is not the case for SAPG, since even when $c=0$, the error
$S_{n+1}^\sa - \barS(\theta_n)$ vanishes as soon as
$\alpha > \beta >0$. Since the case "constant batch size" is the usual
choice of the practitioners in order to reduce the computational cost
of the algorithm, this proposition supports the use of SAPG instead of
MCPG.

We finally study the convergence of SAPG without assuming that the batch
size sequence $\{m_{n}, n \geq 0\}$ is constant, which implies the
following assumption on the sequences
$\{\gamma_n, \delta_n, m_n, n \geq 0\}$.
\begin{assumption} \label{A5} The step size sequences
  $\{\gamma_n, n \geq 0 \}$, $\{\delta_n, n \geq 0 \}$ and the batch
  size sequence $\{m_n, n \geq 0 \}$ satisfy
  \begin{enumerate}[a)]
  \item \label{A5:item1} $\gamma_n \in \ocint{0,1/L}$,
    $\delta_n \in \ooint{0,1}$, $m_n \in \nset$,
    $\sum_n \gamma_n = + \infty$, $\sum_n \gamma_n^2 < \infty$,
\begin{align*}
  &  \sum_{n } \left( \gamma_{n-1} \gamma_{n} + \gamma_{n-1}^2 + |\gamma_n - \gamma_{n-1}| \right)  \D_n < \infty,  \\
  & \sum_n \gamma_n^2 \delta_n^2 (1+\D_{n+1})^2 m_n^{-1} < \infty,
  \end{align*}
  where
  $\D_n \eqdef \sum_{k \geq n} \left( \prod_{j=n}^k
    (1-\delta_j)\right)$.
\item \label{A5:item2} Furthermore,
  \begin{align*}
    &   \sum_n \gamma_{n+1} |m_{n+1}^{-1} \delta_{n+1} - m_n^{-1} \delta_n| < \infty, \\
    &  \sum_n \gamma_{n+1} |m_{n+1}^{-1} \delta_{n+1} \D_{n+2} - m_n^{-1} \delta_n \D_{n+1}| < \infty, \\
    &  \sum_n    \left( \gamma_{n-1} \gamma_{n} + \gamma_{n-1}^2 + |\gamma_n - \gamma_{n-1}| \right)  \cdots \\
    & \qquad \qquad \times m_{n-1}^{-1} \, \delta_{n-1} (1 + \D_{n}) < \infty.
  \end{align*}
  \end{enumerate}
\end{assumption}
 Let us comment this assumption in the case the batch size sequence
$\{m_n, n \geq 0 \}$ is constant. This situation corresponds to the
"stochastic approximation regime" where the number of draws at each
iteration is $ m_n=1$ (or say, $m_n = m$ for any $n$), and it also
corresponds to what is usually done by practitioners in order to
reduce the computational cost. When
$\delta_n= \delta_\star \in \ooint{0,1}$ for any $n \geq 0$, then
$\D_n = \delta_\star^{-1}$ for any $n \geq 0$.  This implies that the
condition H\ref{A5} is satisfied with polynomially decreasing
sequences $\gamma_n \sim \gamma_\star / n^\alpha$ with
$\alpha \in \ocint{1/2,1}$ (and $m_n = m$ for any $n$).

When $\delta_n \sim \delta_\star \, n^{-\beta}$ for
$\beta \in \ooint{0,1}$, then $\D_n = O( n^\beta)$ (see
Lemma~\ref{lem:controle:Dn}). Hence, using
Lemma~\ref{lem:controle:Dn}, H\ref{A5}\ref{A5:item1}) and
H\ref{A5}\ref{A5:item2}) are satisfied with
$\gamma_n \sim \gamma_\star n^{-\alpha}$ where
$\beta < (1+\beta)/2 < \alpha \leq 1$, and $m_n = m$ for any $n$.

We can not have $\delta_n = \delta_\star n^{-1}$ since it implies
$\D_n = + \infty$ for any $n \geq 0$.

\begin{theorem}
  \label{theo:Cvg:SAEMprox}
  Assume H\ref{A1}, H\ref{A1bis}, H\ref{A3} and
  H\ref{A4}\ref{A4:item1}-\ref{A4:item2}). 

  Let $\{\theta_n, n \geq 0 \}$ be the sequence given by
  Algorithm~\ref{algo:SAPG} and applied with sequences
  $\{\gamma_n, \delta_n, m_n, n \geq 0\}$ verifying
  H\ref{A5}\ref{A5:item1}).

  In the biased case, assume also H\ref{A4}\ref{A4:item3}) and
  H\ref{A5}\ref{A5:item2}).
 
  Then with probability one, there exists $\theta_\infty \in \calL$
  such that $\lim_n \theta_n = \theta_\infty$.
\end{theorem} 
\begin{proof}
  The proof is in Section~\ref{sec:proof:SAEMprox}.
\end{proof}

\section{Numerical illustration in the convex case}\label{sec:practical}
  In this section, we illustrate the behavior of the
algorithms MCPG and SAPG on a toy example.  We first introduce the example and
then give some guidelines for a specific choice of the sequences $\{\delta_n, n
\geq 0\}$, $\{\gamma_n, n \geq 0\}$. Finally, the algorithms are compared more
systematically on repeated simulations.

\subsection{A toy example}
\label{sec:algos:toyexample}
The example is a mixed model, where the regression function is linear in the
latent variable $\hid$. More precisely, we observe data $(\obs_1, \cdots,
\obs_N)$ from $N$ subjects, each individual data being a vector of size $J$:
$\obs_k \eqdef (\obs_{k1}, \ldots, \obs_{kJ})$.  For the subject $k$, $k=1,
\cdots, N$, $\obs_{kj}$ is the $j$-th measurement at time $t_{kj}$, $j=1,
\cdots, J$.  It is assumed that $\{\obs_{k}, k=1, \cdots, N\}$ are independent
and for all $k=1, \cdots, N$,
\begin{equation}\label{eq:ModeleLineaireMixte1}
\begin{split}
&Y_{kj} \vert \hid^{(k)} \stackrel{ind}{\sim} \mathcal{N}\left(\pscal{\hid^{(k)}}{\bar t_{kj}}
,1\right), \\ &\bar t_{kj} \eqdef \begin{bmatrix}
  1 \\
  t_{kj}
\end{bmatrix} \qquad j=1, \cdots, J;
\end{split}
\end{equation}
that is, a linear regression model with individual random intercept
and slope, the $\rset^2$-valued vector being denoted by
$\hid^{(k)}$. The latent variable is
$\hid=(\hid^{(1)}, \ldots, \hid^{(N)})$. Furthermore,
\begin{equation}\label{eq:ModeleLineaireMixte2}
 \hid^{(k)} \stackrel{ind}{\sim} \mathcal{N}_2( X_k \theta, I_2);
\end{equation}
here, $\theta \in \rset^{2(D+1)}$ is an
unknown parameter and the design matrix $X_k \in \rset^{2 \times 2(D+1)}$ is known 
\begin{equation}\label{eq:designX}
X_k \eqdef \begin{bmatrix}
 1& X_{k1}& \ldots& X_{k D} &0 & 0& \ldots& 0 \\
 0&  0&  \ldots &0&  1&  X_{k1}& \ldots &X_{k D}
\end{bmatrix}.
\end{equation}
The optimization problem of the form Eq.~(\ref{eq:problem}) that we consider is
the log-likelihood $\ell(\theta)$ penalized by a lasso penalty: the objective
is the selection of the influential covariates
$$(X_{k1}, \ldots, X_{kD})$$ on the two components of $\hid^{(k)}$. We
thus penalize all the elements except $\theta_1$ and $\theta_{D+2}$
which correspond to the two intercepts; hence, we set
\[
g(\theta) \eqdef \lambda \sum_{r \neq \{1, D+2 \}} |\theta_r|.
\]

\medskip

The above model is a latent variable model with complete
log-likelihood equal to - up to an additive constant
\begin{align*}
  - \frac{1}{2}\sum_{k=1}^N\left\{ \sum_{j=1}^{J}  \left( \obs_{kj} - \pscal{\hid^{(k)}}{\bar t_{kj}}\right)^2 +   \left\| \hid^{(k)} - X_k \theta\right\|^2 \right\}.
\end{align*}
It is of the form $\phi(\theta) + \pscal{S(z)}{\psi(\theta)}$ by
setting (with $(\cdot)'$ denoting the transpose of a matrix)
\begin{align*}
\begin{split}
  &\phi(\theta) \eqdef -\frac{1}{2} \theta' \left( \sum_{k=1}^N X_k' X_k\right)
  \theta - \frac{1}{2} \sum_{k=1}^N \sum_{j=1}^J \obs_{kj}^2, \\
  &\psi(\theta) \eqdef \begin{bmatrix}
    1 \\
    \theta
\end{bmatrix} \in \rset^{1+2(D+1)}, \\
&S(z^{(1)}, \cdots, z^{(N)})  \eqdef \\&- \frac{1}{2} \sum_{k=1}^N \begin{bmatrix}
  {z^{(k)}}' (I+T_k) z^{(k)}- 2 \pscal{z^{(k)}}{\bar \obs_k} \\
  - 2 X_k' z^{(k)}
\end{bmatrix}, \\ & T_k \eqdef \sum_{j=1}^{J} \bar t_{kj} \bar t_{kj}',\\ &\bar \obs_k
\eqdef \sum_{j=1}^{J} \obs_{kj} \bar t_{kj}.
\end{split}
\end{align*}
The \textit{a posteriori} distribution $\pi_\theta$ is a Gaussian distribution on
$\rset^{2N}$, equal to the product of $N$ Gaussian distributions on $\rset^2$:
\begin{equation}\label{eq:toyexample:pi}
\begin{split}
&\pi_\theta(z^{(1)}, \cdots, z^{(N)}) \eqdef \\ &\prod_{k=1}^N
\mathcal{N}_2\left((I+T_k)^{-1} \left(\bar \obs_k + X_k \theta \right),  (I+T_k)^{-1} \right)[z^{(k)}].
\end{split}
\end{equation}
Hence, $\barS(\theta)$ is explicit and given by
\begin{equation} \label{eq:toyexample:barS}
\barS(\theta) = -\frac12 \sum_{k=1}^N
\begin{bmatrix}
     \mathrm{Trace}((I+T_k) \Sigma_k)  -2   \bar \obs_k' (I+T_k)^{-1} \left(\bar \obs_k + X_k
      \theta \right)
  \\
  -2  X_k' (I+T_k)^{-1} \left(\bar \obs_k + X_k \theta \right)
\end{bmatrix}
\end{equation}
with
\begin{equation} 
\begin{split}
&\Sigma_k  \eqdef (I+T_k)^{-1}  \\  &+ (I+T_k)^{-1} \left(\bar \obs_k + X_k
  \theta \right) \left(\bar \obs_k + X_k \theta \right)' (I+T_k)^{-1}.
\end{split}
\end{equation}      
Finally, note that in this example, the function $\ell$ is explicit and given by (up to an additive constant) 
\begin{equation*}
\begin{split}
  \ell(\theta) = &- \frac{1}{2} \theta' \left( \sum_{k=1}^N X_k' X_k  \right) \theta \\ &+ \frac{1}{2} \sum_{k=1}^N (\bar \obs_k + X_k \theta)' (I+T_k)^{-1} (\bar \obs_k + X_k \theta).
\end{split}
\end{equation*}
Thus $\ell$ is a concave function. Furthermore, in
  this toy example, $\theta \mapsto \nabla \ell(\theta)$ is linear so
  that the Lipschitz constant $L$ is explicit and equal to
  \begin{equation}\label{eq:Lip:ctt}
  L = \Vert - \sum_{k=1}^N X_k' X_k + \sum_{k=1}^N X_k' (I+T_k)^{-1} X_k
  \Vert_2,
  \end{equation}
 where  for a matrix A, $\| A \|_2$  denotes the spectral norm.   Finally, we assumed that
 $\Theta = \{ \theta \in \mathbbm{R}^{2(D+1)} \vert \Vert \theta \Vert < 10^4 \}$ to fulfill the theoretical boundedness assumption. The MCMC algorithm includes a projection step on $\Theta$ if necessary. But in practice, it never happens.
 \\ \\ 
A data set is simulated using this model with $N=40$, $J=8$, $D=300$
and $t_{kj} \in \{ 0.25,4,6,8,10,12,14,16 \}$,
$\forall k \in \{1,...,N\}$. The design components
$(X_{k1}, \ldots, X_{kD})$ (see Eq.~(\ref{eq:designX})) are drawn from
a centered Gaussian distribution with covariance matrix $\Gamma$
defined by $\Gamma_{rr'} = 0.5^{\vert r - r' \vert}$
($r,r'=1,...,300$).  To sample the observations, we use a parameter
vector $\theta^\star$ defined as follows:
$\theta_1^\star = \theta_{D+2}^\star=1$; the other components are set
to zero, except $12$ components randomly selected ($6$ among the
components $\{2, \cdots, D+1\}$ and $6$ among the components
$\{D+3, \cdots, 2D+2 \}$) and chosen uniformly in $\ccint{0.5,1.5}$ -
see the last row on Figure~\ref{fig:RES_TOY_SUPPORT}.

\subsection{Guidelines for the implementation}
\label{sec:guidelines}

In this section, we give some guidelines on the choice of the sequences
$\{\delta_n, n \geq 0\}$ and $\{\gamma_n, n \geq 0\}$.  We illustrate the
results on  single runs of each algorithm. We use the same random draws for all the algorithms to avoid potential differences due to the randomness of the simulations. Similar results have been
observed when simulations are replicated. We refer to Section~\ref{sec:simuLongTime} for replicated simulations.

Classical sequences $\{\delta_n, n \geq 0\}$ and $\{\gamma_n, n \geq 0\}$ are
of the form:
\begin{equation}
\label{eq:gammaRandom}
 \gamma_{n+1}   = \left\{ 
\begin{array}{ll} 
  \gamma_\star  & \text{ if } n \leq n_\alpha, \\
  \gamma_\star (n-n_\alpha)^{-\alpha}  & \text{ if } n>n_\alpha, \end{array} \right.
\end{equation} 

\begin{equation}
  \label{eq:gammaRandom}
  \delta_{n+1}   = \left\{ 
\begin{array}{ll} 
  \delta_\star  & \text{ if } n \leq n_\beta, \\
  \delta_\star (n-n_\beta)^{-\beta}  & \text{ if } n>n_\beta.
\end{array} \right. 
\end{equation} 

 \begin{figure*}
\begin{center}
\includegraphics[scale = 0.55]{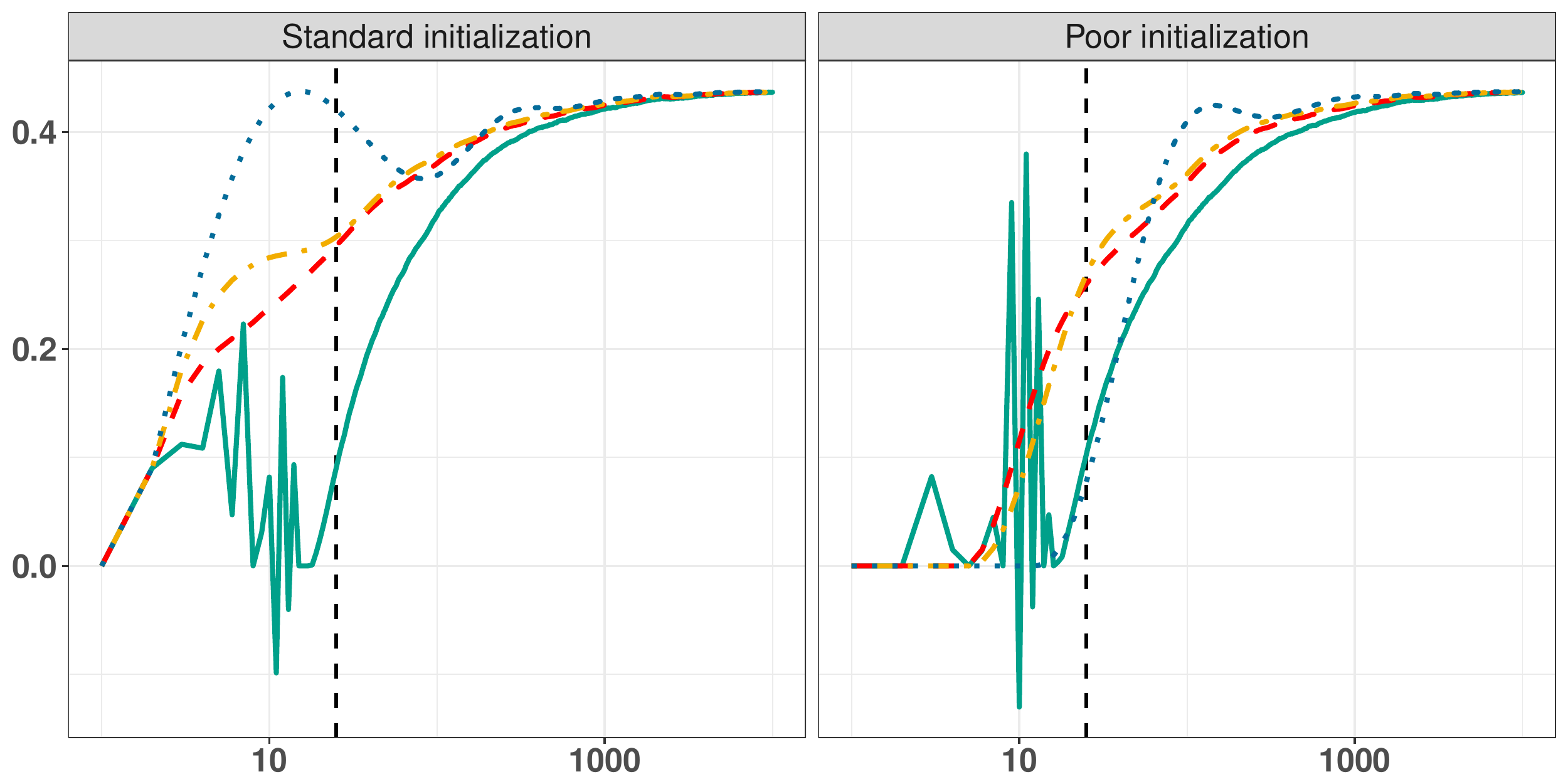}
\caption{Estimation of the component $\# 245$ of the
  vector $\theta$ along $10000$ iterations ($x$-axis is in $\log_{10}$ scale)
  for MCPG and SAPG. MCPG is represented in green solid line. A run of SAPG is
  displayed in dashed red line in the case $\delta_\star=0.8$; in
  dashed-dotted yellow line in the case $\delta_\star=0.5$ and in dotted blue line in the case $\delta_\star=0.2$. For all runs,
  $\gamma_{\star}=0.009$, $n_\alpha = n_\beta = 0$, $(\alpha, \beta) = (0.75,
  0.5)$ and $m_n = 60$. The vertical black dashed line corresponds to the smallest $n$ such
  that $\gamma_{n} \leq 1/L$. Left: standard initialization of $\theta$
  ($\theta_{n=0}$ is the null vector).  Right: poor initialization of $\theta$. The penalty term $\lambda$ was set to 50 for all the runs.  }
\label{fig:GUIDE}
\end{center}
\end{figure*}

\textbf{Impact of $\gamma_\star$ and $\delta_\star$ on the transient phase}:
the theoretical study on the asymptotic behavior of SAPG and MCPG is derived
under the assumption that $\gamma_n \leq 1/L$: when $\alpha >0$, this property
holds for any $n$ large enough. In this section, we illustrate the role of
$\gamma_n, \delta_n$ for small values of $n$ that is, in the transient phase of
the algorithm.  In Figure~\ref{fig:GUIDE}, we display the behavior of MCPG and
SAPG for two different values of the initial point $\theta_{n=0}$: on the left, it
corresponds to a standard initialization ($\theta_{n=0} = (0, \cdots, 0)$) while on
the right, it corresponds to a poor initialization - which mimics what may
happen in practice for challenging numerical applications.

On both plots, we indicate by a vertical line the smallest $n$ such that
$\gamma_{n} \leq 1/L$ - remember that in this example, $L$ is explicit
(see Eq.~\eqref{eq:Lip:ctt}).  The plots show the estimation of
component \#245, as a function of the number of iterations $n$. In all cases,
$n_\alpha = n_\beta =0$, $\alpha = 0.75$, $m_n= 60$, and for SAPG, $\beta =
0.5$. The dotted blue curve displays a run of SAPG when $(\gamma_\star,\delta_\star) = (0.009,0.2)$; the dashed-dotted yellow curve displays a run of SAPG when
$(\gamma_\star,\delta_\star) = (0.009,0.5)$; the dashed red curve displays a
run of SAPG when $(\gamma_\star,\delta_\star) =(0.009,0.8)$; the green solid
curve displays a run of MCPG when $\gamma_\star = 0.009$.

The stability of MCPG during the transient phase  depends crucially on the first
values of the sequence $\{\gamma_n, n \geq 0\}$.  Then when $n$ is large enough
so that $\gamma_{n} \leq 1/L$ (after the vertical line), MCPG is   more
stable and gets smoother. For SAPG, a 
small value of $\delta_\star$ implies an important impact of the initial point $\theta_{n=0}$. 
When this initial point is poorly chosen, a small value of $\delta_\star$  delays 
the convergence of SAPG. A value of $\delta_\star$ around $0.5$ is   a good compromise.

 \textbf{Role of ${\alpha}$ and ${\beta}$}: Figure \ref{fig:GUIDE_ALPHA}
displays the behavior of SAPG for different values of $\alpha$ and $\beta$ with
$(\gamma_\star, \delta_\star)=(0.015, 0.5)$, $n_\alpha = n_\beta=0$ and
$m_n=60$. The plots show that the larger the parameter $\alpha$ is, the longer
the transient phase is. We then recommend to set $\alpha$   close to $0.6$. 
The parameter $\beta$ seems to have an impact
only when $\alpha$ is close to $1$. Therefore, we recommend to set $\delta_n$
constant during the transient phase ($n_\beta >0$) and then to   decrease it
  rapidly   in the convergence phase.

 \begin{figure*}
\begin{center}
\includegraphics[scale = 0.55]{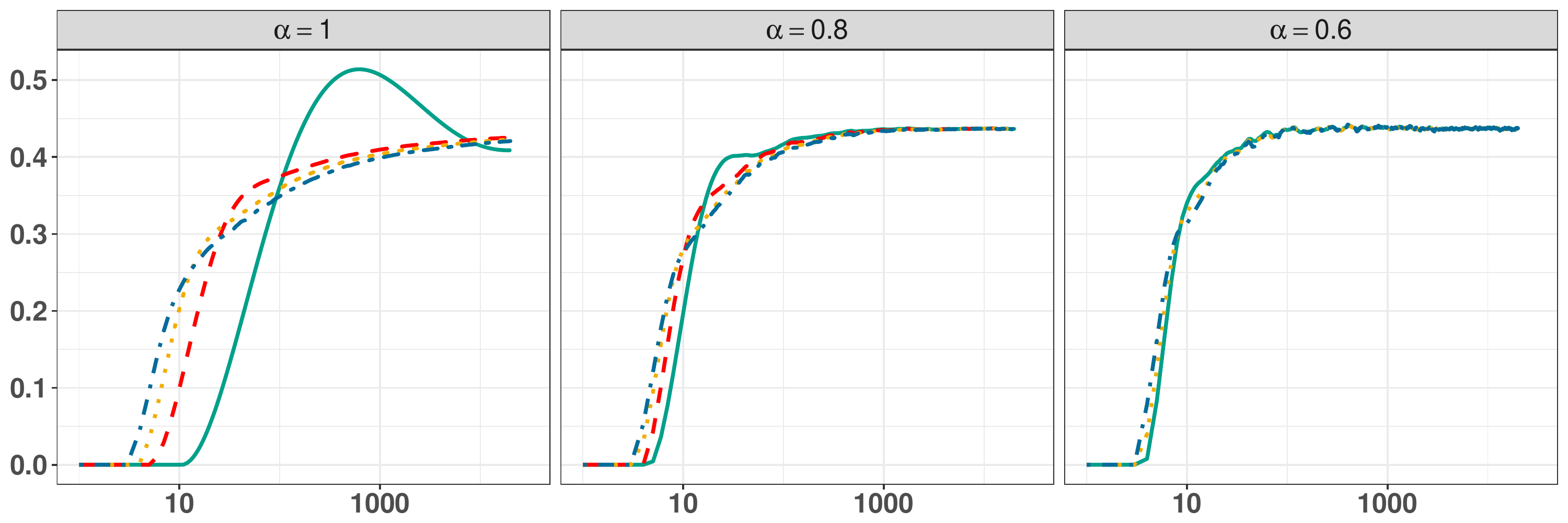}
\caption{Estimation of the component $\# 245$ of the
  vector $\theta$ along $10000$ iterations ($x$-axis is in $\log_{10}$ scale).
  Runs of SAPG with $(\gamma_\star, \delta_\star)=(0.015,0.5)$, $n_\alpha =
  n_\beta=0$ and $m_n=60$ and different values of $(\alpha, \beta)$. (Left)
  $\alpha = 1$ and $\beta=0.9$ (green solid line), $\beta=0.6$ (red dashed
  line), $\beta=0.3$ (yellow dotted line), $\beta=0$ (blue dash-dotted line).
  (Middle) $\alpha = 0.8$ and $\beta=0.5$ (green solid line), $\beta=0.3$ (red
  dashed line), $\beta=0.1$ (yellow dotted line), $\beta=0$ (blue dash-dotted
  line).  (Right) $\alpha = 0.6$ and $\beta=0.2$ (green solid line),
  $\beta=0.1$ (yellow dotted line), $\beta=0$ (blue dash-dotted line). The penalty term $\lambda$ was set to 50 for all the runs. }
\label{fig:GUIDE_ALPHA}
\end{center}
\end{figure*}

 \textbf{Random stepsize sequence $\{\gamma_n, n \geq 0\}$}: The convergence of
the SAPG algorithm can suffer from the scale difference of the parameters, when
run with the same stepsize sequence $\{\gamma_n, n \geq 0\}$ applied to each
component of $\theta_n$.

 Ideally each component of $\theta_n$ should have a specific $\gamma_n$
value adapted to its scale. But it can be time-consuming to find, by hand-tuning, a sequence
that ensures a fast and stable convergence of the algorithm.  As an
alternative, we suggest to use a matrix-valued random sequence
$\{\Gamma^n, n\geq 0\}$ and replace the update rule of SAPG by
 $$ (\theta_{n+1})_i  = \Prox_{ \Gamma^{n+1}_{ii} g} \left( (\theta_n)_i +
    \Gamma^{n+1}_{ii} \left( \phi(\theta_n) + \Psi(\theta_n)
      S_{n+1}^\sa\right)_i \right). 
$$
 We propose to define the matrix
$\Gamma_{n+1}$ as a diagonal matrix with entries
$ \Gamma^{n+1}_{ii}$ depending on $H^n_{ii}$,
where $H^n$ is an approximation of the hessian of the likelihood
$\ell(\theta)$ (we give an example of such an approximation in
Section~\ref{sec:applications}).  Through numerical
  experiments, we observed that asymptotically, $H^n$
  converges. Hence, to ensure a stepsize sequence decaying like
  $O(n^{-\alpha})$ asymptotically, we propose the following definition
of the random sequence:
\begin{equation}
\label{eq:gammaRandom}
 \Gamma^{n+1}_{ii}  = \left\{ 
\begin{array}{ll} 
  1/H^n_{ii} & \text{ if } n \leq n_0, \\  \left((n-n_0)^{\alpha} H^n_{ii} \right)^{-1} & \text{ if } n>n_0. \end{array} \right.
  \end{equation}

\subsection{Long-time behavior of the algorithm}\label{sec:simuLongTime}

In this section, we illustrate numerically the theoretical results on the long term 
convergence of the algorithms MCPG, SAPG and SAEM-pen  (\textit{i.e.}  Algorithm \ref{algo:EM} applied
with $S_{n+1} = S_{n+1}^\sa$) and EM-pen on the toy model. In this
example, the exact algorithm EM-pen (see Eq.~\eqref{eq:EM:iteration}) applies:
the quantity $\barS(\theta)$ is an explicit expectation under a Gaussian
distribution $\pi_\theta$. Therefore, we use this example \textit{(i)} to
illustrate the convergence of the three stochastic methods to the same limit
point as EM-pen, \textit{(ii)} to compare the two approximations $S_{n+1}^\mc$
and $S_{n+1}^\sa$ of $\barS(\theta_n)$ in a GEM-pen approach, and
\textit{(iii)} to study the effect of relaxing the M-step by comparing the
GEM-pen and EM-pen approaches namely SAPG and SAEM-pen.

The sequences $\{\gamma_n, n \geq 0 \}$ and $\{\delta_n, n \geq 0 \}$ are
defined as follows: $(\gamma_\star, \delta_\star) = (0.004, 0.5)$, and
$n_\alpha = n_\beta = 0$; three different pairs $(\alpha, \beta)$ are
considered: $(\alpha, \beta)=(0.9, 0.4)$, $(\alpha, \beta)= (0.6,0.1)$, and
$(\alpha, \beta)=(0.5,0.5)$.  The algorithms are implemented with a fixed batch
size $m_n=60$.  $100$ independent runs of each algorithm are performed.  For
the penalty term, we set $\lambda=50$.  In MCPG, SAPG and SAEM-pen, the
simulation step at iteration $(n+1)$ relies on exact sampling from
$\pi_{\theta_n}$ - see Eq.~\eqref{eq:toyexample:pi}; therefore, in this toy
example, the Monte Carlo approximation of $\barS(\theta_n)$ is unbiased.

On Figure~\ref{fig:Stat_suff}, for the three algorithms MCPG, SAPG and
SAEM-pen, the evolution of an approximation of $\| S_{n+1} - \bar
S(\theta_n)\|_{2}$ with iterations $n$ is plotted, where, for a random variable
$U$, $\|U\|_2 \eqdef \sqrt{\PE\left[ \|U \|^2 \right]}$.  This $L_2$-norm is
approximated by a Monte Carlo sum computed from $100$ independent realizations
of $S_{n+1}$; here, $\barS(\theta_n)$ is explicit (see
Eq.~\eqref{eq:toyexample:barS}).  SAEM-pen and SAPG behave similarly; the
$L_2$-norm converges to $0$, and the convergence is slower when
$(\alpha,\beta)=(0.6,0.1)$ - this plot illustrates the result stated in
Proposition~\ref{prop:L2norm:DeltaS}, Section~\ref{sec:theorie}.  This
convergence does not hold for MCPG because the size $m_n$ of the Monte Carlo
approximation is kept fixed.

\begin{figure*}
\begin{center}
\includegraphics[scale = 0.45]{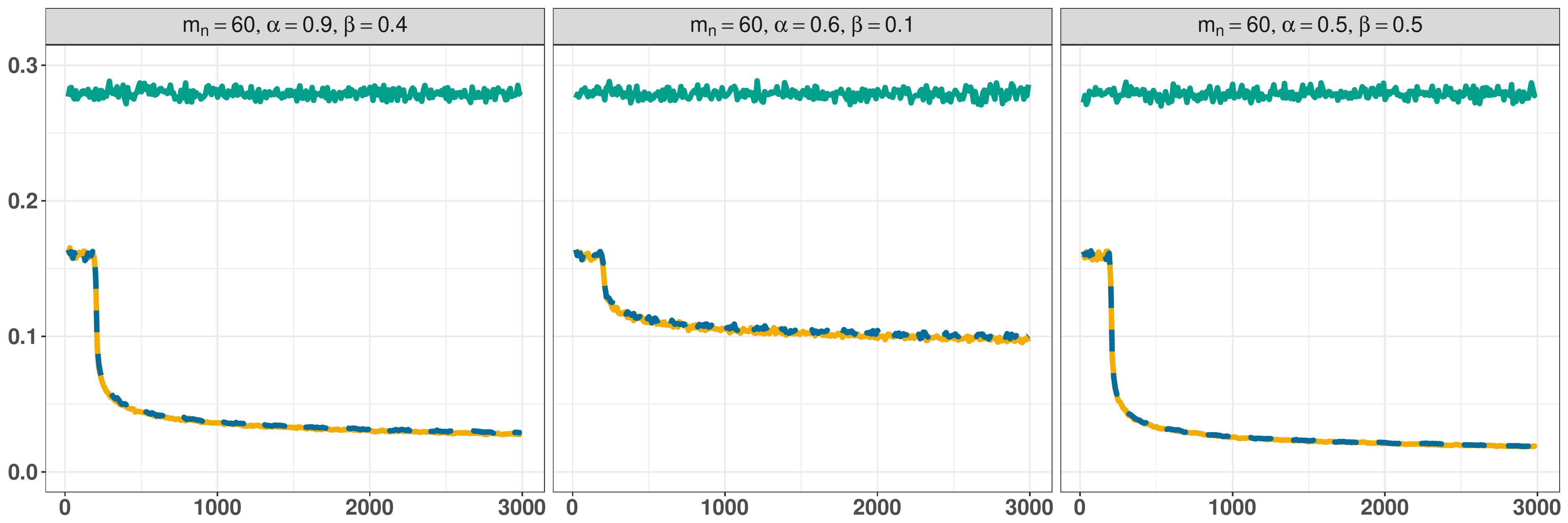}
\caption{Evolution of the Monte Carlo approximation of
  $\| S_{n+1} - \bar S(\theta_n)\|_{2}$ with iterations $n$ for
  algorithms MCPG (solid green), SAEM-pen (solid yellow), SAPG (dashed
  blue), implemented with $(\alpha,\beta) = (0.9, 0.4)$ [left],
  $(\alpha, \beta) = (0.6, 0.1)$ [center] and
  $(\alpha, \beta) = (0.5, 0.5)$ [right]; for MCPG and SAPG, the batch
  size is fixed $m_n = 60$.}
\label{fig:Stat_suff}
\end{center}
\end{figure*}

We compared the limiting vectors $\lim_n \theta_n$ obtained by each
algorithm, over the $100$ independent runs. They are all equal, and
the limiting vector is also the limiting value $\theta_\infty$ of the
EM-pen algorithm. 
In order to discuss the rate of convergence, we show the behavior of the
algorithms when estimating the component $\# 245$ of the regression
coefficients; this component was chosen among the non-null component of
$\theta_\infty$. Figure~\ref{fig:RES_TOY_ITER} shows the boxplot of $100$
estimations of the component $\# 245$ of the vector $\theta_n$, when $n = 5,
25, 50, 500$, for the algorithms MCPG, SAPG and SAEM-pen with $(\alpha, \beta)
= (0.9, 0.4)$. Here, SAPG and MCPG behave similarly, with a smaller variability
among the $100$ runs than SAEM-pen.  SAEM-pen converges faster than SAPG and
MCPG which was expected since they correspond respectively to stochastic
perturbations of EM-pen and GEM-pen algorithms.  Figure~\ref{fig:RES_TOY} shows
the boxplot of $100$ estimations by MCPG, SAPG and SAEM-pen of the component
$\# 245$ after $n=500$ iterations with different values for the parameters
$\alpha$ and $\beta$. We observe that the three algorithms give similar final estimates for the three conditions on parameters $\alpha$ and $\beta$. This is due to the fact that with
 $n_\alpha = n_\beta = 200$, the algorithms
have already attained the convergence phase when $n=200$. This allows the algorithms to quickly converge  toward the limit points when $n>200$.

\begin{figure}
\begin{center}
\includegraphics[scale = 0.45]{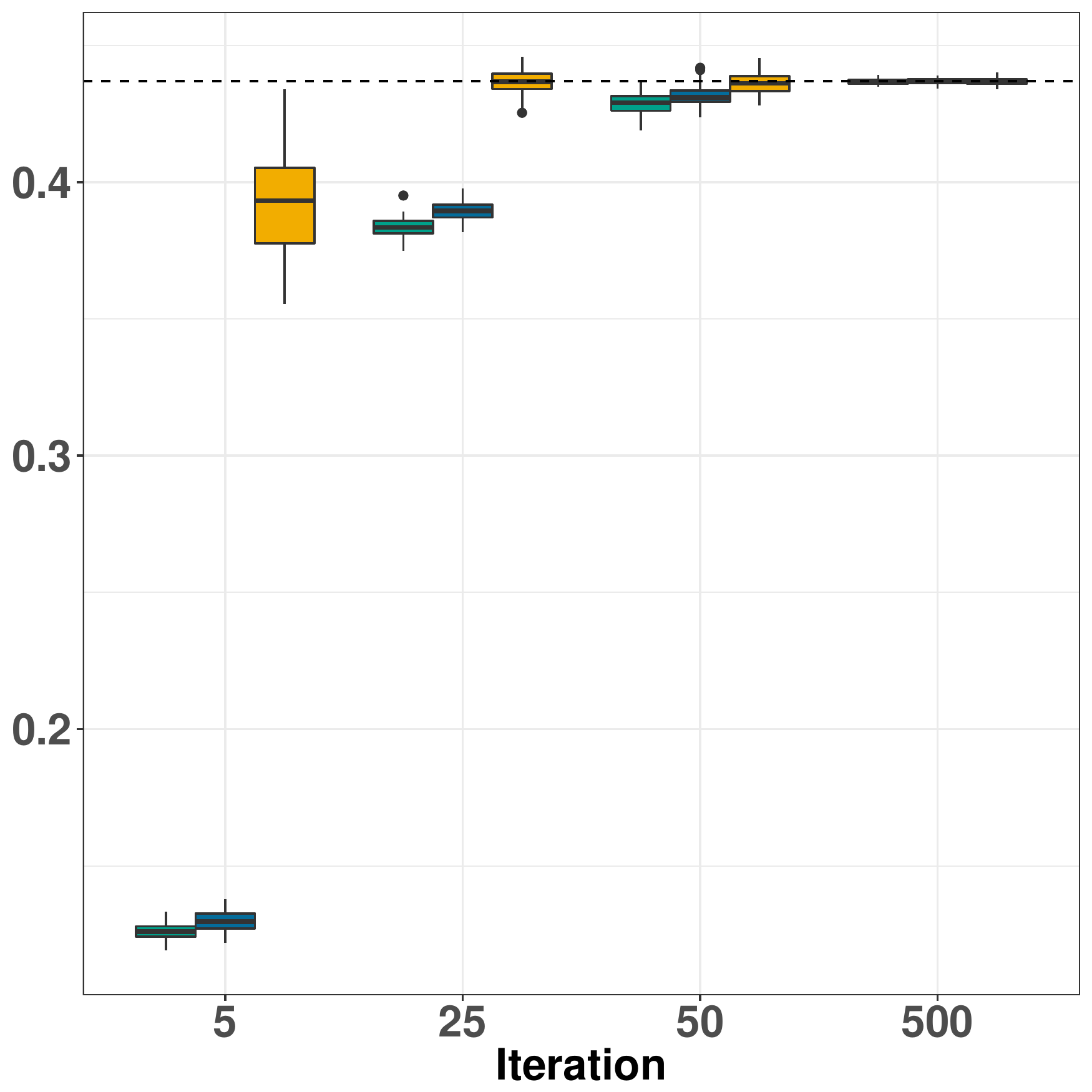}
\caption{Estimation of the component $\# 245$ of the vector $\theta_n$
  when $n=5, 25, 50, 500$.  MCPG (green, left), SAPG (blue, middle)
  and SAEM-pen (yellow, right) are implemented with
  $(\alpha,\beta) = (0.9, 0.4)$; for MCPG and SAPG, the batch size is
  fixed $m_n = 60$. Each boxplot is computed from $100$ independent
  runs. Black dashed line correspond to the value obtained with EM-pen
  algorithm at iteration $500$. }
\label{fig:RES_TOY_ITER}
\end{center}
\end{figure}

\begin{figure*}
  \begin{center}
\includegraphics[scale = 0.35]{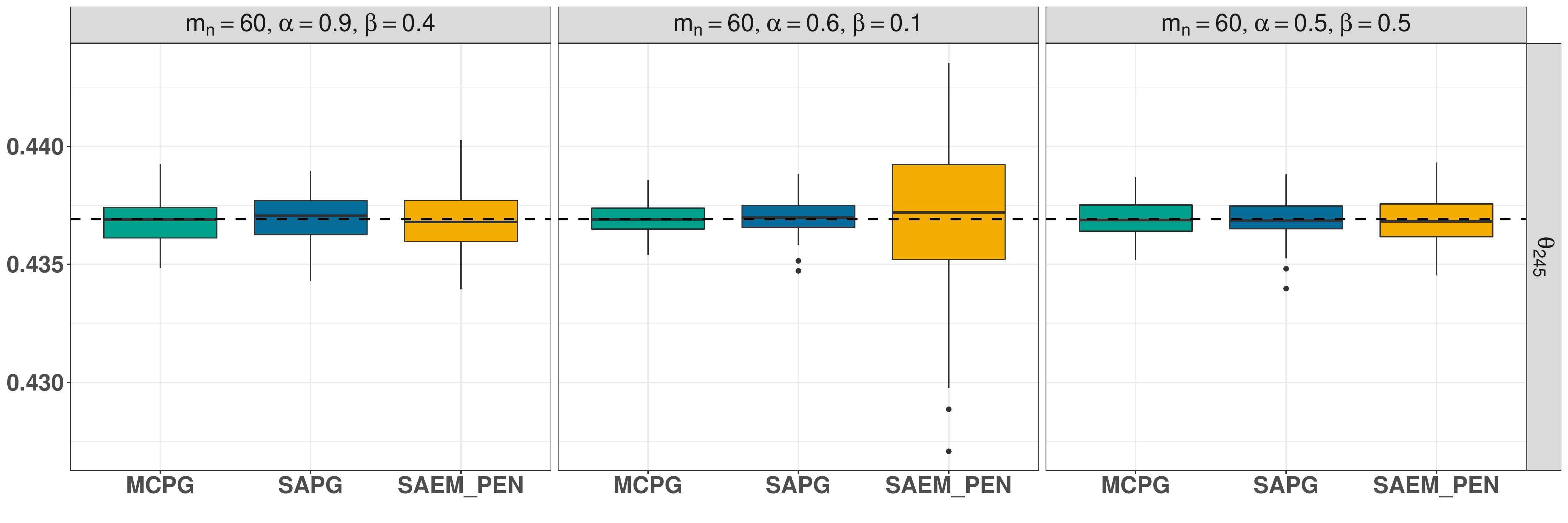}
\caption{Estimation of the component $\# 245$ of the vector $\theta_n$
  when $n=500$.  MCPG (green), SAPG (blue) and SAEM-pen (yellow) are
  implemented with $(\alpha,\beta) = (0.9, 0.4)$ [left],
  $(\alpha, \beta) = (0.6, 0.1)$ [center] and
  $(\alpha, \beta) = (0.5, 0.5)$ [right]; for MCPG and SAPG, the batch
  size is fixed $m_n = 60$. Each boxplot is computed from $100$
  independent runs. Black dashed line correspond to the value obtained
  with EM-pen algorithm at iteration $500$. }
\label{fig:RES_TOY}
\end{center}
\end{figure*}

Figure~\ref{fig:LOG_LIK} shows the convergence of a Monte Carlo
approximation of $n \mapsto \PE\left[F(\theta_n)\right]$ based on
$100$ independent estimations $\theta_n$ obtained by three different
algorithms: EM-pen, MCPG, SAPG and SAEM-pen run with
$(\alpha, \beta)=(0.9,0.4)$ and $m_n = 60$. Here again, all the
algorithms converge to the same value and EM-pen and SAEM-pen converge
faster than MCPG and SAPG. We observe that the path of SAPG is far
more smooth than the path of MCPG.
\begin{figure}
\begin{center}
\includegraphics[scale = 0.5]{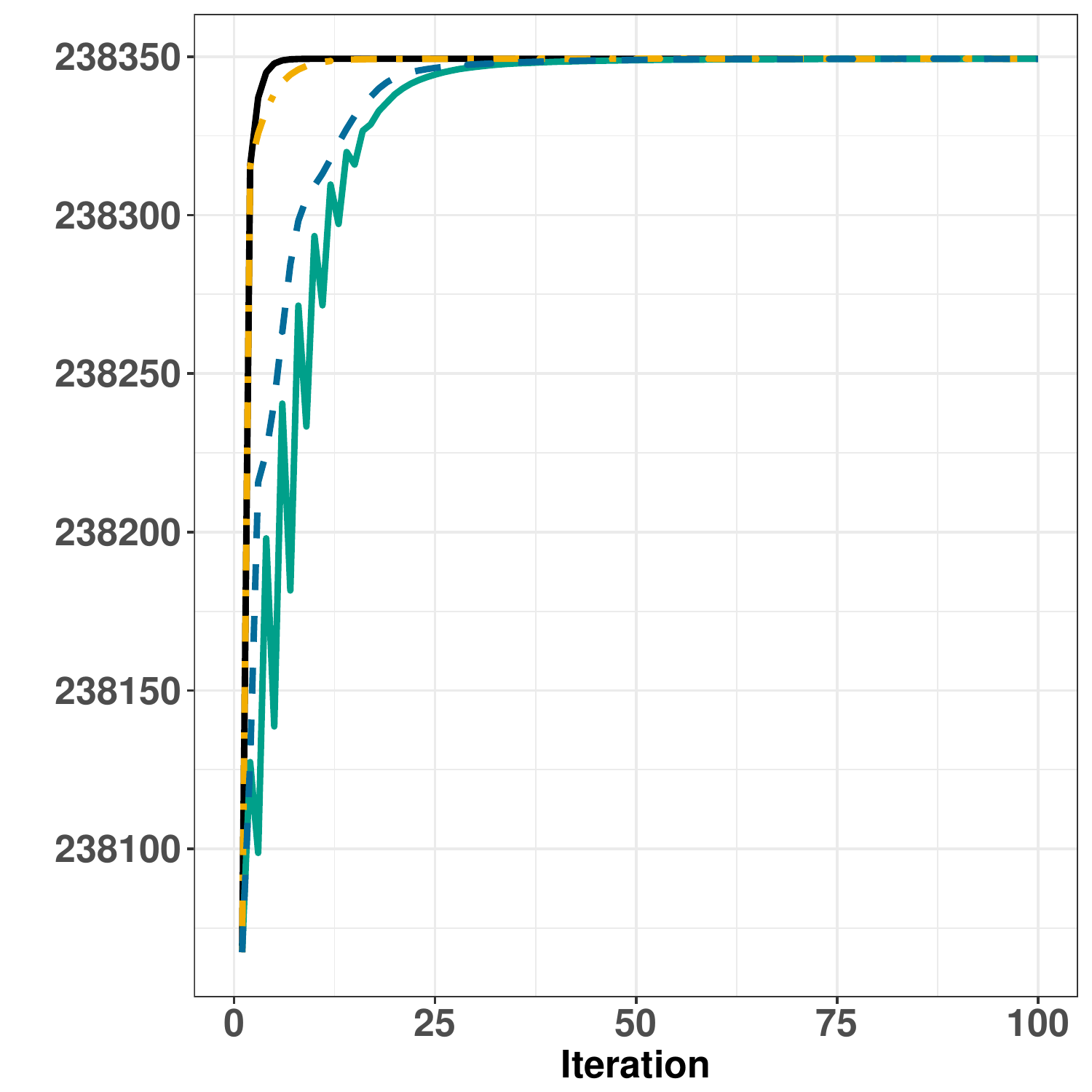}
\caption{Monte Carlo approximation of $\PE\left[F(\theta_n)\right]$
  (based on $100$ independent samples) along the iterations $n$, for
  algorithms EM-pen (solid black), MCPG (solid green), SAEM-pen
  (dash-dotted yellow), SAPG (dashed blue), implemented with
  $(\alpha,\beta)=(0.9,0.4)$ and $m_n=60$.}
\label{fig:LOG_LIK}
\end{center}
\end{figure}

Finally, Figure~\ref{fig:RES_TOY_SUPPORT} shows the support of the
vector $\lim_n \theta_n$ (where the component $\theta_1$ and
$\theta_{302}$ are removed) estimated by MCPG, SAPG, SAEM-pen and
EM-pen (the estimated support is the same for the four algorithms). The frequency, among $100$ independent
runs, for each component to be in the support of the limit value
$\lim_n \theta_n$, is displayed.  Algorithms are implemented with
$(\alpha, \beta) = (0.9,0.4)$ and $m_n = 60$.  For all algorithms, we
observe that most of the non-null components of $\lim_n \theta_n$ are
non-null components of $\theta^\star$. Note also that the stochastic
algorithms MCPG, SAPG and SAEM-pen converge to the same vector as
EM-pen.
\begin{figure*}
\begin{center}
\includegraphics[scale = 0.65]{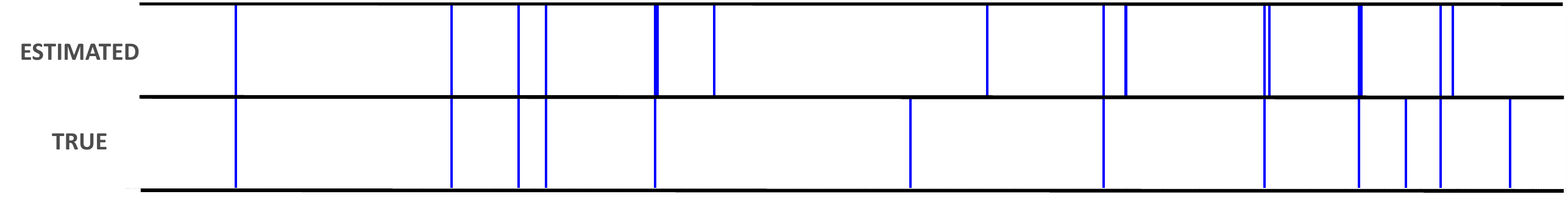}
\caption{(top) Support of $\lim_n \theta_n$ estimated by all the
  algorithms MCPG, SAPG, SAEM-pen and EM-pen over $100$ runs for
  $(\alpha,\beta)=(0.9,0.4)$ and $m_n = 60$. (bottom) The support of
  $\theta^\star$ used to produce the observations. On both rows, the components $1$
  and $D+1$ are not displayed.}
\label{fig:RES_TOY_SUPPORT}
\end{center}
\end{figure*}

\section{Inference in non-linear mixed models for pharmacokinetic data}
\label{sec:applications}
In this section, SAPG is applied to solve a more challenging
problem. The objective is to illustrate the algorithm in cases that are not covered by the theory. 
The application is in pharmacokinetic analysis, with non-linear mixed
effect models (NLMEM);  in this application, the
  penalized maximum-likelihood inference is usually solved by the
  SAEM-pen algorithm, possibly combined with an approximation of the
  M-step when it is non explicit. This section also provides a
  numerical comparison of SAPG and SAEM-pen. Both algorithms have a
  simulation step; in this more challenging application, it will rely
  on a Markov chain Monte Carlo (MCMC) sampler - see
  Section~\ref{sec:MCMC:PK}. Therefore, for both algorithms,
  $\barS(\theta)$ is approximated by a biased Monte Carlo sum.

We start with a presentation of the statistical analysis and its
translation into an optimization problem; we then propose a
modification of the SAPG by allowing a random choice of the stepsize
sequence $\{\gamma_n, n \geq 0\}$, to improve the numerical properties
of the algorithm.   We conclude the section by a
  comparison of the methods on a pharmacokinetic real data set.

\subsection{The non-linear mixed effect model}
Pharmacokinetic data are observed along time for $N$ patients.  Let
$\obs_{k}$ be the vector of the $J$ drug concentrations observed at
time $t_{kj}$ ($j \in \{1,\ldots,J \}$) for the $k$-th patient
($k \in \{1,\ldots,N\}$). The kinetic of the drug concentration is
described by a non-linear pharmacokinetic regression model $f$, which
is a function of time $t$ and unobserved pharmacokinetic parameters
$\hid^{(k)}$. These parameters are typically the rates of absorption
or elimination of the drug by the body. An example is detailed
below. The variability among patients is modeled by the randomness of
the hidden variables $\hid^{(k)}$. These pharmacokinetic parameters
may be influenced by covariates, such as age, gender but also genomic
variables. Among these high dimension factors, only few of them are
correlated to $\hid^{(k)}$. Their selection can thus be performed by
optimizing the likelihood with a sparsity inducing penalty, an
optimization problem that enters problem Eq.~(\ref{eq:problem}). However,
the likelihood is generally not concave, that is, through this
example, we explore beyond the framework in which we are able to prove
the convergence of MCPG and SAPG (see Section~\ref{sec:theorie}).
 
 Let us now detail the model and the optimization problem.  The mixed model is   defined as
 \begin{equation}\label{eq:mixedmodel0}
   \obs_{kj} = f(t_{kj},\hid^{(k)}) +  \epsilon_{kj} , \quad \quad
   \epsilon_{kj} \sim \mathcal{N}(0,\sigma^2) \mbox{ (iid)}, 
 \end{equation}
 where the measurement errors $\epsilon_{kj}$ are centered,
 independent and identically normally distributed with variance
 $\sigma^2$. Individual parameters $\hid^{(k)}$ for the $k$-th subject
 is a $R$-dimensional random vector, independent of
 $\epsilon_{kj}$. In a high dimension context, the $\hid^{(k)}$'s
 depend on covariates (typically genomics variables) gathered in a
 matrix design $X_k\in \mathbbm{R}^{R \times (D+1)R}$. The
 distribution of $\hid^{(k)}$ is usually assumed to be normal with
 independent components 
\begin{equation} \label{eq:mixedmodel} 
  \hid^{(k)} \stackrel{ind}{\sim}
  \mathcal{N}_R(X_{k} \mu ,\Omega)
 \end{equation}
 where
 $\mu \in \mathbbm{R}^{(D+1)R}$ is the mean parameter vector and
 $\Omega$ is the covariance matrix of the random parameters
 $\hid^{(k)}$, assumed to be diagonal.  The unknown parameters are
 $\theta = \left(\mu,  \Omega_{11}, \cdots,
     \Omega_{RR}, \sigma^2 \right)\in  \mathbb{R}^{R(D+1)} \times \ooint{0,+\infty}^{R+1}$.

 A typical function $f $ is the two-compartmental pharmacokinetic
 model with first order absorption, describing the distribution of a
 drug administered orally. The drug is absorbed from the gut and
 reaches the blood circulation where it can spread in peripheral
 tissues. This model corresponds to $f= \frac{A_c}{V_c}$ with $A_c$
 defined as
\begin{eqnarray} \label{eq:A1C2}
  \frac{dA_{d}}{dt}& =& -k_a \, A_{d},\nonumber\\
  \frac{dA_{c}}{dt} &=& k_a \, A_{d} + \frac{Q}{V_{p}}A_{p}  - \frac{Q}{V_{c}}A_{c}- \frac{Cl}{V_{c}}A_{c} , \\
  \frac{dA_{p}}{dt} &= & \frac{Q}{V_{c}}A_{c} -  \frac{Q}{V_{p}}A_{p} \nonumber,  
\end{eqnarray}
with $A_{d}(0) = Dose$, $A_{c}(0)  = 0$, $A_{p}(0) = 0$ and 
where $A_d$, $A_c$, $A_p$ are the amount of drug in the depot, central
and peripheral compartments, respectively; $V_c$, $V_p$ are the volume
of the central compartment and the peripheral compartment,
respectively; $Q$ and $Cl$ are the inter compartment and global
elimination clearances, respectively. To assure positiveness of the
parameters, the hidden vector is
$$z=(\log(V_c), \log(V_p), \log(Q), \log(Cl), \log(k_a)).$$
 
It is easy to show that the model described by
 Eqs.~\eqref{eq:mixedmodel0}-\eqref{eq:mixedmodel} belongs to the curved
exponential family (see Eq.~\eqref{eq:intro:latent}) with minimal sufficient
statistics:
\begin{eqnarray}
  && S_{1k}(z) = z^{(k)},  \quad S_{2}(z) = \sum_{k=1}^N z^{(k)} \,
    z^{(k)'},  \nonumber \\ 
  &&S_{3}(z) = \sum_{k=1}^N\sum_{j=1}^{J}  (\obs_{kj}-f(t_{kj}, z^{(k)}))^2;  \nonumber \\
  && \psi_{1k}(\theta) = ( X_{k}  \mu)'\Omega^{-1}, \quad \psi_2(\theta) = -\frac{1}{2}\Omega^{-1},  \nonumber \\
  && \psi_3(\theta) = - \frac{1}{2\sigma^2}, \nonumber
\end{eqnarray}
and $S(z) \eqdef \mathrm{Vect}\left(S_{11}(z), \cdots, S_{1N}(z), S_2(z),
  S_3(z) \right)$, $\psi \eqdef \mathrm{Vect}\left(\psi_{11}, \cdots,
  \psi_{1N},\psi_2, \psi_3\right)$.  The function $\phi$ is given by
$\phi(\theta) = -J N \log(\sigma ) - \frac{N}{2}\log(\vert \Omega \vert) -
\frac{1}{2}\sum_{k} ( X_{k} \mu)'\Omega^{-1}( X_{k} \mu)$.

The selection of genomic variables that influence all coordinates of
$\hid^{(k)}$ could be obtained by optimizing the log-likelihood
penalized by the function $g(\theta)= \lambda \Vert \mu \Vert_{1}$,
the $L_1$ norm of $\mu$ with $\lambda$ a regularization parameter.

However, this estimator is not invariant under a scaling transformation (ie $\tilde \hid^{(k)} = b\hid^{(k)} \mbox{, } \tilde \mu = b\mu \mbox{ and } \tilde \Omega_{rr}^{1/2} = b\Omega_{rr}^{1/2}$) (see \textit{e.g.} \citep{lehmann2006theory}).
In our high dimension experiments, the scale of the hidden variables has a non negligible influence on the selection of the support. To be more precise, let us
denote, for $r \in \{1, \ldots, R\}$,
$$\mu_{(r)} \eqdef (\mu_{(r-1)(D+1)+1}, \ldots, \mu_{r(D+1)})$$
 the
coordinates corresponding to the $r$-th pharmacokinetic parameter of
function $f$. When the variance $\Omega_{rr}$ of the random parameters
$Z_r^{(k)}$ is low, the algorithms tend to select too many covariates.
This phenomenon is strengthened with a small number of subjects as
random effect variances are more difficult to estimate.  A solution is
to consider the
following penalty
 $$\lambda \sum_{r=1}^{R} \Omega_{rr}^{-\frac{1}{2} }  \Vert \mu_{(r)} \Vert_{1} ,$$
 that makes the estimator invariant under scaling transformation. It
 was initially proposed by \cite{stadler2010} to estimate the
 regression coefficients and the residual error's variance in a
 mixture of penalized regression models. However, the resulting
 optimization problem is difficult to solve directly because the
 variance of the random effect $\Omega_{rr}$ appears in the penalty
 term. Therefore, we propose a new parameterization
 \[
 \tilde{\mu}_{(r)} \eqdef \mu_{(r)}\Omega_{rr} ^{-\frac{1}{2} }, \qquad
 \qquad \Sigma_{rr} \eqdef \Omega_{rr} ^{-\frac{1}{2} }
 \]
 and $\tilde{\theta}
 \eqdef \{ \tilde{\mu},  \Sigma_{11}, \cdots, \Sigma_{RR}, \sigma^2
 \} \in \rset^{R(D+1)} \times \ooint{0,+\infty}^{R+1}$.  Then, the optimization problem is the following:
\begin{gather}\label{eq:ScaledPenLL2}
  \underset{\tilde{\theta}} {\operatorname{Argmax}} \left(
    \ell(\tilde\theta) - g(\tilde \theta) \right), \quad \mbox{ with }
  g(\tilde \theta)= \lambda \Vert \tilde{\mu} \Vert_{1}.\end{gather}
This problem can be solved using MCPG, SAPG or SAEM-pen
algorithms. Indeed, the complete log-likelihood is now - up to an
additive constant -
 \begin{equation*}
 \begin{split}
   \log p(\obs,Z ; \tilde\theta)= &- J N \log(\sigma
   ) \\ &-\frac{1}{2}\sum_{k=1}^N \sum_{j=1}^J \frac{ \left( Y_{kj} -
       f(t_{kj}, \hid^{(k)} )\right)^{2} }{\sigma^2} \\ &+ N\log(\vert
   \Sigma \vert) - \frac{1}{2}\sum_{k=1}^N  \| \Sigma \hid^{(k)} - X_{k}
   \tilde{\mu}\|^2
\end{split}
\end{equation*}
 It is again a complete likelihood from the exponential
  family, with the   statistic $S$ unchanged and the functions $\phi$
and $\psi$ given by  - up to an additive constant - 
\begin{eqnarray*}
  &&\phi(\tilde\theta) =  -J N \log(\sigma ) + N\log(\vert \Sigma \vert)  - \frac{1}{2}\sum_{k=1}^N\| X_{k}  \tilde\mu \|^2,  \\
  &&\psi_{1k}(\tilde\theta) = \Sigma ( X_{k}  \tilde\mu)^{t}, \quad  \psi_2(\tilde\theta) =  -\frac{1}{2}\Sigma^{2} ,  \quad   \psi_3(\tilde\theta) =  - \frac{1}{2\sigma^2}. 
\end{eqnarray*}
With these definitions of $\phi, \psi$ and $g$, the M-step of SAEM-pen
amounts to compute the optimum of a convex function, which is solved
numerically by a call to a cyclical coordinate descent implemented in
the R package {\sf glmnet} \citep{Friedman2010}.

\paragraph{MCMC sampler.}
\label{sec:MCMC:PK}
In the context of non-linear mixed models, simulation from
$\pi_{\theta_n}\rmd \nu$ can not be performed directly like in the toy
example. We then use a MCMC sampler based on a Metropolis Hastings
algorithm to perform the simulation step. Two proposal kernels are
successively used during the iterations of the Metropolis Hastings
algorithm. The first kernel corresponds to the prior distribution of
$\Sigma\hid^{(k)}$ that is the Gaussian distribution
$\mathcal{N}(X_{k} \tilde \mu_n , I)$.  The second kernel corresponds
to a succession of $R$ uni-dimensional random walk in order to update
successively each component of $\hid^{(k)}$. The variance of each
random walk is automatically tuned to reach a target acceptance ratio
following the principle of an adaptive MCMC algorithm
\citep{andrieu2008tutorial}.

 \paragraph{Adaptive random stepsize sequences.}
 In the context of NLMEM, numerical experiments reveal that choosing a
 deterministic sequence $\{\gamma_n, n\geq 0\}$ that achieve a fast convergence
 of SAPG algorithm could be difficult.  Indeed, parameters to estimate are of
 different scales. For example, random effect and residual variances are
 constrained to be positive. Some of them are close to zero, some are not.
As explained in Section \ref{sec:guidelines}, an alternative is to 
  implement a matrix-valued random sequence $\{\Gamma^n, n\geq 0\}$. 
%
The gradient
and the hessian of the likelihood $\ell(\theta)$ can be approximated
by stochastic approximation using the Louis principle
\citep[see][Chapter 4]{McLachlan:Krishnan:2008}. Let us denote $H_n$
the stochastic approximation of the hessian obtained at iteration $n$
as explained by  \cite{samson2007saem}. Note that no
supplementary random samples are required to obtain this
approximation.  
Along the iterations, each diagonal entry of the matrix
$ {H}^{n} $ converges: this limiting value can be seen as a
simple way to automatically tune a good $\gamma_{\star}$, that is
parameter specific.  The entries
$\Gamma^{n+1}_{ii}$  are then defined by Eq.~(\ref{eq:gammaRandom}). 


\subsection{Simulated data set.}
The convergence of the corresponding algorithms is illustrated on
simulated data.  Data are generated with the model defined by Eq.~\eqref{eq:A1C2} and
$N = 40$, $J = 12$, $D = 300$. The design matrix $X_k$ is defined by
Eq.~\eqref{eq:designX}, with components $(X_{k1}, \ldots, X_{kD})$
drawn from $\mathcal{N}(0,\Gamma)$ with
$\Gamma_{ii'} = 0.5^{\vert i - i' \vert}$ ($i,i'=1,...,300$).
Parameter values are
\begin{equation}
\begin{split}
& [\mu_1, \mu_{1+(D+1)} , \mu_{1+2(D+1)} , \mu_{1+3(D+1)} ,
\mu_{1+4(D+1)} ] \\ &=  [ 6.61 , 6.96, 5.77, 5.42 ,-0.51]; \nonumber
\end{split}
\end{equation}
the other components are set to zero, except $\mu_{4}$ and $\mu_{912}$
that are set to $1$. The matrix $\Omega$ is diagonal with diagonal
elements equal to $(0.16, 0.16, 0.16, 0.04, 0.04)$.

The penalty function is set to
\begin{equation}\label{eq:penaltyNLMEM}
g(\tilde \theta) \eqdef \lambda \sum_{\ell \neq \{1+r(D+1), r=0, \cdots, 4  \}} |\tilde\mu_\ell |,
\end{equation}
only the parameters corresponding to a covariate effect being
penalized.  The optimization problem Eq.~\eqref{eq:problem} with
regularization parameter $\lambda = 190$ is solved on this dataset
with SAEM-pen and SAPG; we run SAPG with the random sequence
$\{\Gamma_n, n \geq 0\}$ as described above (see
\eqref{eq:gammaRandom}) with $n_0 = 9500$. For both algorithms, the
stochastic approximation step size was set to:
  \begin{equation}
    \delta_{n+1}  = \left\{ \begin{array}{ll} 0.5  & \text{ if } n\leq n_0 \\  \frac{0.5}{(n - n_0)^{\beta}} & \text{ if } n>n_0 \end{array} \right.
\end{equation}
We set $\alpha = 0.75$ and $\beta = 0.499$. Figure~\ref{fig:HD_SIM_PK}
  shows the convergence of SAEM-pen and
three parameterizations of SAPG: i) a version with $\gamma^\star=0.005$ for all the components of $\theta$, ii) a version with $\gamma^\star=0.005$ for $\tilde{\mu}$, $\gamma^\star=0.0005$ for $\Sigma$ and $\gamma^\star=0.03$ for $\sigma$, and iii) a version with adaptive random step sizes. For the four algorithms,  all the parameters corresponding to a
covariate effect are estimated to zero except the two components  $\mu_{4}$ and $\mu_{912}$. The version of SAPG with a same $\gamma^\star$ for all the component is the one that converge the most slowly. When the $\gamma^\star$ is tuned differently according the type of parameters, the convergence of SAPG is accelerated. Algorithms SAEM-pen and SAPG with adaptive random step sizes  have a similar fast convergence profile.

Figure~\ref{fig:HD_SIM_PK_StepSize} presents the evolution of four
entries of the matrix $\Gamma^n$ along the iterations of SAPG,
corresponding to the components $\tilde{\mu}_{904}$,
$\tilde{\mu}_{912}$, $\Sigma_{44}$ and $\sigma$. We can notice that
they are not on the same scale. They vary during the first iterations
and converge to limiting values before iteration $n_0=9500$. Then the
step sizes decrease to 0, following the definition given in
Eq.~\eqref{eq:gammaRandom}.

 \begin{figure*}
\begin{center}
\includegraphics[scale = 0.55]{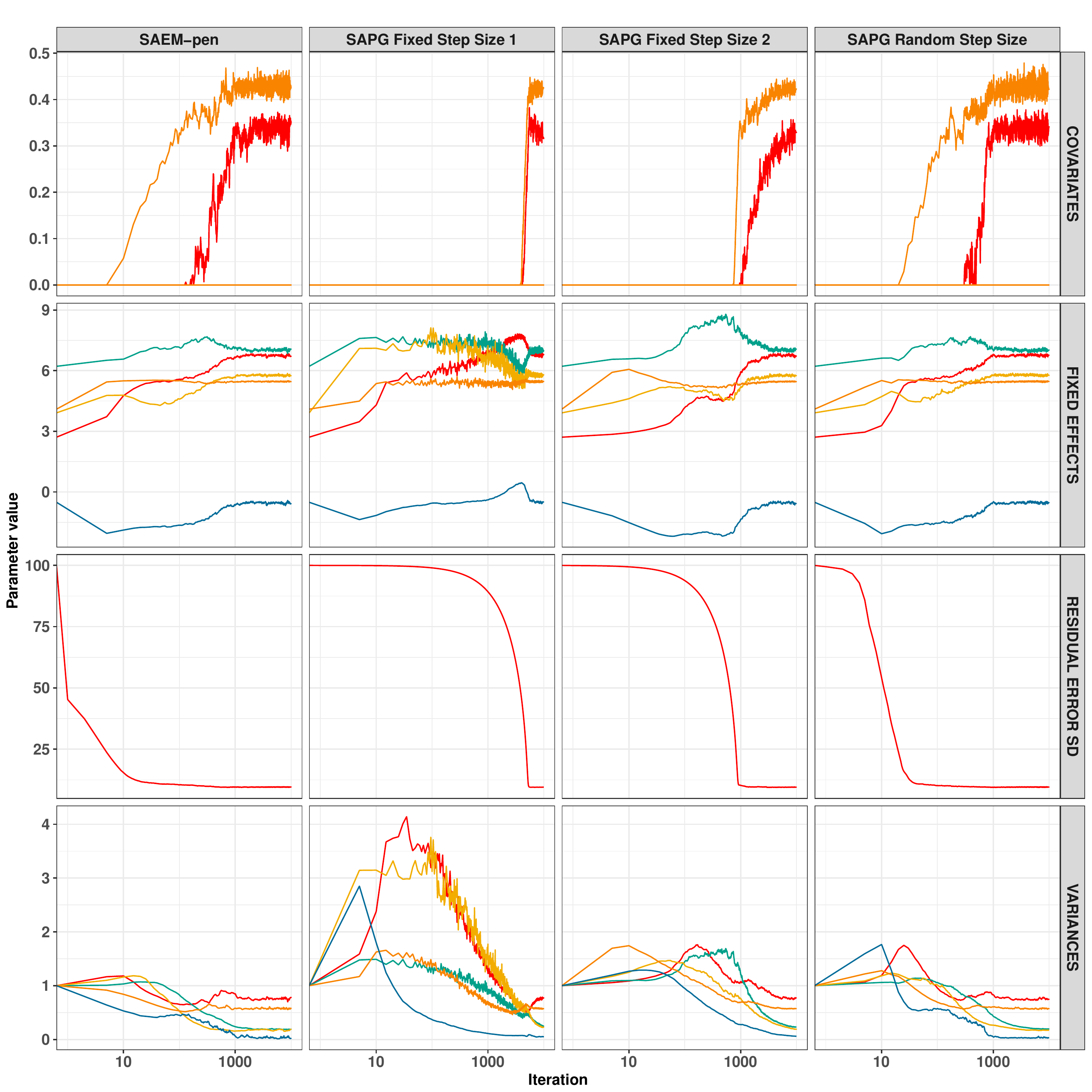}
\caption{Path of a run of SAEM-pen [left column] and three different parameterizations of SAPG : i) with $\gamma^\star=0.005$ for all the components of $\theta$ [middle left column], ii) with $\gamma^\star=0.005$ for $\tilde{\mu}$, $\gamma^\star=0.0005$ for $\Sigma$ and $\gamma^\star=0.03$ for $\sigma$  [middle right column] and iii) with
  a random sequence $\{\Gamma^n,n \geq 0 \}$ [right column].
  For each algorithm, estimation of the standard
    deviation of the residual error $\sigma$ [third row];
  the variances of the $Z^{(k)}$'s,
    $\Omega_{11}, \cdots, \Omega_{RR}$ [fourth row];
  the path of the covariate parameters $\mu_i$ for
    $i \notin \{1, 1+(D+1), \cdots, 1+4(D+1) \}$ [first row];
  the path of the intercept parameters
    $\mu_{i}, i \in \{ 1, 1+(D+1), \cdots, 1+4(D+1)\}$ [second row]. Each color corresponds to a specific parameter: orange line for
  $Cl$, red for $V_c$, blue line for $k_a$, yellow line for $Q$ and green line for
  $V_p$. Note that the path of all the covariate parameters is zero except for two components. $x$-axis is in $\log_{10}$ scale.}
\label{fig:HD_SIM_PK}
\end{center}
\end{figure*}

 \begin{figure*}
\begin{center}
\includegraphics[scale = 0.275]{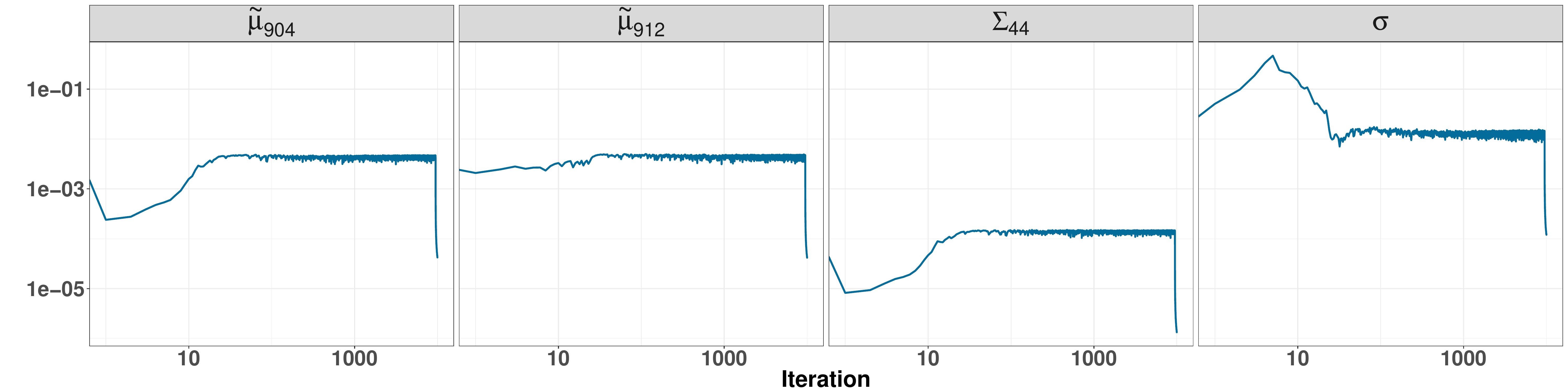}
\caption{Evolution of $\Gamma^n_{ii}$ with iterations $n$ of SAPG, for four
  different values of $i$, corresponding to the components
  $\tilde{\mu}_{904}$ [left]; $\tilde{\mu}_{912}$ [middle left];  
  $\Sigma_{44}$ [middle right]; $\sigma$ [right]. Both $x$-axis and
  $y$-axis are in $\log_{10}$ scale.}
\label{fig:HD_SIM_PK_StepSize}
\end{center}
\end{figure*}

\subsection{Application to real data}

Algorithms SAEM-pen and SAPG with matrix-valued random sequence
$\{\Gamma^n, n\geq 0\}$ are applied to real data of the
pharmacokinetic of dabigatran ($DE$) from two cross over clinical
trials \citep{delavenne2013semi,ollier2015vitro}. These $2$ trials
studied the drug-drug interaction between $DE$ and different
Pgp-inhibitors. From these $2$ trials, the pharmacokinetics of $DE$
are extracted from $15$ subjects with no concomitant treatment with
Pgp-inhibitors. The concentration of dabigatran is measured at $9$
sampling times for each patient. Each subject is genotyped using the
DMET$\textregistered$ microarray from Affymetrix. Single nucleotide
polymorphisms (SNP) showing no variability between subjects are
removed and $264$ SNP are included in the analysis.

Function $f$ of the non-linear mixed model is defined as the two
compartment pharmacokinetic model with first order absorption
previously described (see Eq.~\eqref{eq:A1C2})
\citep{delavenne2013semi}. The penalty function $g$ is defined by
Eq.~\eqref{eq:penaltyNLMEM}.

Because of the limited number of subjects, the influence of genetic
covariates is only studied on $V_c$ and $Cl$ parameters, that
characterize the elimination process and are the most likely to be
influenced by the genetic.  Finally, random effect variances of $Q$
and $V_p$ are set to $0.01$ in accordance with previously published
population pharmacokinetic of dabigatran
\citep{delavenne2013semi}. The other variance parameters are
estimated.  The penalized likelihood problem (Eq.~\ref{eq:ScaledPenLL2})
is solved on the data with the SAEM-pen and SAPG algorithms, for $40$
different values of parameter $\lambda$. SAPG algorithm is run using
the random sequence $\{\Gamma^n, n\geq 0\}$ given in
Eq.~\eqref{eq:gammaRandom}. The best regularization parameter
$\lambda$ is chosen with a data-driven approach based on the EBIC
criteria \citep{chen2008extended}.

\begin{figure}
\begin{center}
\includegraphics[scale = 0.325]{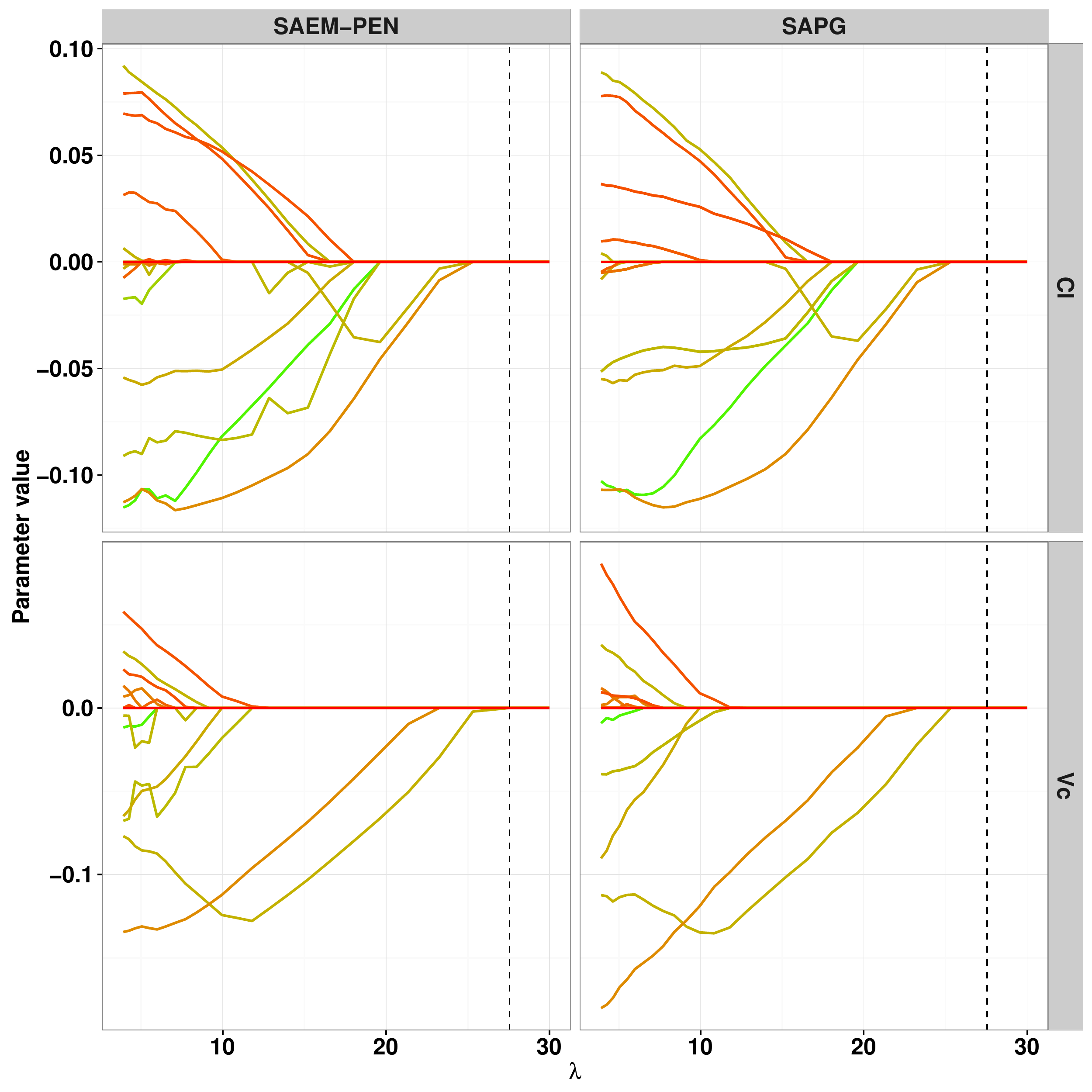}
\caption{ Regularization path of covariate parameters ($Cl$
  parameter on top, $V_c$ parameter on bottom) obtained on dabigatran
  pharmacokinetic data for both SAEM-pen and SAPG algorithms. Black
  vertical dashed line corresponds to the $\lambda$ value selected by
  $EBIC$.}
\label{fig:ScaledRegPathDABI}
\end{center}
\end{figure}

Figure~\ref{fig:ScaledRegPathDABI} shows the results. The
regularization paths of $Cl$ and $V_c$ parameters using both algorithms
correspond to the evolution of covariate coefficient estimates as a
function of the value of $\lambda$. They are reconstructed with low
noise for both algorithms, are very similar for high values of
$\lambda$ but less for lower values of $\lambda$. 

Finally, the selected model has all covariates parameters set to
zero. This means that none of the genetic covariates influence the
distribution of the individual parameters. This result is not
surprising given the low number of subjects and the fact that a large
part of the inter individual variability is due to the dissolution
process of the drug \citep{ollier2015vitro} and is therefore not influenced by genetic
covariates. This lack of relationship between dabigtran's pharmacokinetic parameters and genetic covariates has already been highlighted in an other study \citep{gouin2016inter}.

\section{Conclusion}
In this work, we propose a new Stochastic Proximal-Gradient algorithm
to solve penalized maximum likelihood problems when the likelihood is
intractable: the the gradient is approximated through a stochastic
approximation scheme. We provide a theoretical convergence analysis of
this new algorithm and illustrate these results numerically on a
simulated toy example in the case of a concave likelihood
function. The robustness to the non concave case is explored through a
more challenging application to population pharmacokinetic analysis
relying on penalized inference in non-linear mixed effects models.

\bigskip

\appendix
\textbf{\large{Appendix}}
\section{Proof of Proposition~\ref{prop:monotonicity:EMprox}}
\label{sec:proof:EMprox}
\begin{lemma}\label{lem:Taylor1:Qem}
  Under the assumptions of Proposition~\ref{prop:monotonicity:EMprox}, for any $\gamma \in \ocint{0,1/L}$,
  $s \in \barS(\Theta)$ and any $\theta, \theta' \in \Theta$,
\begin{align}\label{eq:QEM:majoration}
  \Q(\theta \vert \theta') & \geq \Q(\theta' \vert \theta')  \nonumber \\
& - \frac{1}{2\gamma}\|\theta -
  \theta' - \gamma \{ \nabla \phi(\theta') + \jacob{\psi}{\theta'}
  \barS(\theta')\}\|^2  \nonumber \\
& + \frac{\gamma}{2} \|\nabla \phi(\theta') + \jacob{\psi}{\theta'}
  \barS(\theta') \|^2.
\end{align}
\end{lemma}
\begin{proof}
  Fix $\theta' \in \Theta$ and $s \in \barS(\Theta)$. The derivative of the
  function $\theta \mapsto \mathcal{L}(\theta) \eqdef \phi(\theta) +
  \pscal{s}{\psi(\theta)}$ is $\nabla \phi(\theta) + \jacob{\psi}{\theta} s$
  and this gradient is $L$-lipschitz.  From a Taylor expansion to order $1$ at
  $\theta'$ and since the gradient is Lipschitz, we have
\[
\mathcal{L}(\theta) \geq \mathcal{L}(\theta') + \pscal{\nabla
  \phi(\theta') + \jacob{\psi}{\theta'} s }{\theta-\theta'} -
\frac{L}{2} \|\theta- \theta'\|^2.
\]
We then choose $s = \barS(\theta')$, use $L \leq 1/\gamma$ and
conclude by the equality
$2 \pscal{a}{b} - \|a\|^2= \|b\|^2-\|a-b\|^2$.
\end{proof}

\paragraph{Proof of Proposition~\ref{prop:monotonicity:EMprox}}
We prove that
$\Q(\theta_{n+1} \vert \theta_n) -g(\theta_{n+1}) \geq \Q(\theta_n
\vert \theta_n) - g(\theta_n)$
so that the sequence $\{\theta_n, n \geq 0 \}$ defined by
Eq.~\eqref{eq:PG:iteration} is a sequence
satisfying Eq.~\eqref{eq:GEM:iteration}.

By Lemma~\ref{lem:Taylor1:Qem}, it holds for any $\theta \in \Theta$ and any
  $\gamma \in \ocint{0, 1/L}$
  \begin{multline*}
    \Q(\theta \vert \theta_n) - g(\theta) \\
    \geq \Q(\theta_n \vert \theta_n) + \frac{\gamma}{2} \|\nabla
    \phi(\theta_n) + \jacob{\psi}{\theta_n} \barS(\theta_n) \|^2  \\
    - \frac{1}{2\gamma} \left\|\theta - \theta_n - \gamma \{ \nabla
      \phi(\theta_n) + \jacob{\psi}{\theta_n} \barS(\theta_n)
      \}\right\|^2 -g(\theta).
  \end{multline*}
  Note that the RHS and the LHS are equal when $\theta = \theta_n$ so
  that for any point $\tau$ which maximizes the RHS, it holds
  $\Q(\tau \vert \theta_n) - g(\tau) \geq \Q(\theta_n \vert \theta_n)
  - g(\theta_n) $.
  This concludes the proof upon noting that such a point $\tau$ is
  unique and equal to $\theta_{n+1}$ given by Eq.~\eqref{eq:PG:iteration}.

\section{Technical lemmas} 
Define
\begin{eqnarray*} \label{eq:Suite:Delta} && \Delta_{k:n} \eqdef
  \prod_{j=k}^n (1-\delta_j), \ \ 0 \leq k \leq n, \quad
                                          \Delta_{n+1:n} =1,  \\
                                        && \D_k \eqdef \sum_{n \geq k}
                                          \Delta_{k:n}.  \label{eq:Suite:Ddroit}
\end{eqnarray*}
\begin{lemma} \label{lem:ResultOnDelta} For any $n \geq 2$,
  $\sum_{j=2}^{n}\Delta_{j+1:n} \, \delta_j = 1 - \Delta_{2:n}$.
\end{lemma}
\begin{proof}
  For any $j \leq n$, we have
  $\Delta_{j+1:n} - \Delta_{j:n} = \delta_j \Delta_{j+1:n}$ from which
  the result follows.
\end{proof}

\begin{lemma} \label{lem:controle:Dn}
  Let $\beta \in \ooint{0,1}$ and $\delta_\star>0$. Set
  $\delta_n = \delta_\star n^{-\beta}$ for any $n \geq 1$. Then for any $k$ large enough,
\[
\delta_k \D_k \leq 1+ O\left( k^{\beta-1}\right).
\]
Furthermore,
$\left| \delta_{n+1} \D_{n+2} - \delta_n \D_{n+1} \right| =
O(1/ n^{1+(1-\beta) \wedge \beta})$.
\end{lemma}
The proof of Lemma~\ref{lem:controle:Dn} relies on standard Taylor's
expansions with explicit formulas for the remainder. The proof is
omitted.

\begin{lemma} \label{lem:Vitesse:Sn} Let $\beta \in \coint{0,1}$ and
  $\delta_\star>0$. For any $r$, when $n \to \infty$,
  \[
  \sum_{j=2}^n j^{-r} \prod_{k=j}^n \left(1 -
    \frac{\delta_\star}{k^\beta} \right) = O\left( n^{\beta -r}
  \right).
  \]
 \end{lemma}
 \begin{proof}
   We have
\begin{align*}
  \prod_{k=j}^n \left(1 - \frac{\delta_\star}{k^\beta} \right) & \leq
                                                                 \exp\left( -\delta_\star \sum_{k=j}^n k^{-\beta} \right) \\
                                                               & \leq
                                                                 \exp\left( - \frac{\delta_\star}{1-\beta} \left\{n^{1-\beta} -
                                                                 j^{1-\beta} \right\} \right).
\end{align*}
Let $q_\star \geq 0$ such that for any $q \geq q_\star$, $q(1-\beta)+1-r >0$.
For any constant $D>0$, there exist constants $C, C'$ (whose value can
change upon each appearance) such that
\begin{align*}
  & \sum_{j=2}^n j^{-r} \exp(D j^{1-\beta}) \\
&  =  \sum_{q \geq 0}  \frac{D^q}{q!}  \sum_{j=2}^n j^{-r+q(1-\beta)}  \\
&\leq C  n^{1-r} \sum_{q \geq q_\star} \frac{D^q}{q!} \frac{n^{q(1-\beta)}}{(q+1)(1-\beta)} \frac{(q+1)(1-\beta)}{q(1-\beta) +1-r}  + C' n \\
&  \leq C n^{\beta-r} \sum_{q \geq q_\star} \frac{D^q}{(q+1)!} n^{(q+1)(1-\beta)} + C' n   \\
&\leq C \exp(D n^{1-\beta}) n^{\beta-r}.
\end{align*}
This concludes the proof.
 \end{proof}

\begin{lemma}
  \label{lem:IterateApproxSto}
  Let $\{\A_n, n\geq 0 \}$ be a sequence of $d' \times q$ matrices and
  $\{\sigma_n, n \geq 0 \}$ be a sequence of $q \times 1$ vectors. Let
  $\{S_n^\sa, n \geq 0 \}$ be given by Eq.~\eqref{eq:SAApprox}. For any
  $n \geq 2$
\begin{align*}
  \A_{n}\left( S_{n}^\sa - \sigma_{n-1} \right) & =  \Delta_{2:n} \A_1 \left( S_1^\sa -\sigma_0 \right)  \\
&+ \sum_{j=2}^{n} \Delta_{j:n} \left( \A_j - \A_{j-1} \right) \left( S_{j-1}^\sa - \sigma_{j-2} \right) \\
  & + \sum_{j=2}^{n} \Delta_{j:n} \A_j \left( \sigma_{j-2} - \sigma_{j-1}
  \right)  \\
& \hspace{-1.5cm} + \sum_{j=2}^{n} \Delta_{j+1:n} \delta_j \A_j \left(m_{j}^{-1} \sum_{k=1}^{m_j} S(\hid_{k,j-1})-
    \sigma_{j-1} \right).
\end{align*}
\end{lemma}
\begin{proof}
  By definition of $S_n^\sa$, it holds
  $\A_n \left( S_n^\sa - \sigma_{n-1} \right) = (1-\delta_n)
  \A_{n-1} \left( S_{n-1}^\sa - \sigma_{n-2}\right) +B_n$ where
\begin{align*}
  B_n & \eqdef (1-\delta_{n}) (\A_n - \A_{n-1}) \left(S_{n-1}^\sa -
        \sigma_{n-2} \right)  \\
      &+ (1-\delta_n) \A_n \left( \sigma_{n-2} -
        \sigma_{n-1} \right)  \\
&+ \delta_n \A_n \left( m_{n}^{-1} \sum_{k=1}^{m_n} S(\hid_{k,n-1}) -
        \sigma_{n-1} \right).
\end{align*}
By iterating, we have
$$\A_n \left( S_n^\sa - \sigma_{n-1}\right) = \Delta_{2:n} \A_1 \left(S_1^\sa
  - \sigma_0 \right) + \sum_{j=2}^n \Delta_{j+1:n} B_j,$$
from which the lemma follows.
\end{proof}

\begin{lemma} \label{lem:IteratesKernels} Assume
  H\ref{A4}\ref{A4:item1}).  Let $\{S_n^\sa, n \geq 0 \}$ be given
  by Eq.~\eqref{eq:SAApprox}. Then
  \begin{align*}
   & \sup_{n \geq 0} \PE\left[ \| m_n^{-1} \sum_{j=1}^{m_n} S(\hid_{j,n-1}) \|^2  \right] < \infty, \\
& \sup_{n \geq 0} \PE\left[   m_n^{-1} \sum_{j=1}^{m_n} W(\hid_{j,n-1})  \right] < \infty, \\
  &  \sup_{n \geq 0} \PE\left[ \|S_n^\sa \|^2 \right] < \infty.
\end{align*}
\end{lemma}
\begin{proof}
  By H\ref{A4}\ref{A4:item1}), there exists a constant $C < \infty$
  such that for any $n \geq 1$ and $1 \leq j \leq m_n$,
  $ \|S(\hid_{j,n-1})\|^2 \leq C \, W(\hid_{j,n-1})$.  In addition, by the
  drift assumption on the kernels $P_\theta$, we have
\begin{align*}
  \PE\left[W(\hid_{j,n-1}) \right] &= \PE\left[P_{\theta_{n-1}} W(\hid_{j-1,n-1}) \right]   \\
  &\leq \lambda \PE\left[ W(\hid_{j-1,n-1}) \right] + b  \\
  &\leq \lambda^j \PE\left[W(\hid_{0,n-1}) \right]+ \frac{b}{1-\lambda}.
\end{align*}
Similarly, by using $\hid_{0,n-1} = \hid_{m_{n-1},n-2}$, we have
\[
\PE\left[W(\hid_{0,n-1}) \right] \leq \lambda^{m_{n-1}}
\PE\left[W(\hid_{0,n-2}) \right] + \frac{b}{1-\lambda}.
\]
A trivial induction shows that
$$\sup_n \sup_{j \leq m_n} \PE\left[W(\hid_{j,n-1}) \right] < \infty,$$
from which the first two results follow.  For the third one: by
Lemma~\ref{lem:IterateApproxSto} applied with $\A_n = I$ (the identity
matrix) and $\sigma_n = 0$, we have for any $n \geq 1$,
\[
S_{n}^\sa = \Delta_{2:n} S_1 + \sum_{j=2}^{n} \Delta_{j+1:n} \,
\delta_{j} \, m_j^{-1} \sum_{k=1}^{m_j} S(\hid_{k,j-1}).
\]
By the Minkowsky inequality and the inequality
$(a+b)^2 \leq 2 a^2 + 2 b^2$, we have 
\begin{multline*}
  \PE\left[\| S_{n}^\sa \|^2 \right] \leq 2 \left( \Delta_{2:n}
  \right)^2 \PE\left[ \|S_1^\sa \|^2 \right]  \\
  + 2 \sup_j \PE\left[ \|m_j^{-1} \sum_{k=1}^{m_j} S(\hid_{k,j-1})
    \|^2 \right] \, \left( \sum_{j=2}^{n} \Delta_{j+1:n} \, \delta_{j}
  \right)^2.
\end{multline*}
By definition, $\Delta_{2:n} \in \ccint{0,1}$ and by
Lemma~\ref{lem:ResultOnDelta},
$$\sup_n \sum_{j=2}^{n} \Delta_{j+1:n} \, \delta_{j} < \infty.$$ Hence,
$\sup_{n \geq 0} \PE\left[ \|S_n^\sa \|^2 \right] < \infty$.
\end{proof}

 Define the Proximal-Gradient operator
\[
T_\gamma(\theta) \eqdef \Prox_{\gamma,g}\left( \theta + \gamma \nabla \ell(\theta) \right).
\]

\begin{lemma}
  \label{lem:DeltaTheta} Assume H\ref{A1}, H\ref{A1bis} and
  H\ref{A3}. Let $\{S_n^\sa, n \geq 0 \}$ be given
  by~\eqref{eq:SAApprox}. Then, for the sequence
  $\{\theta_n,n \geq 0 \}$ given by Algorithm~\ref{algo:SAPG},
  \begin{enumerate}[(i)]
  \item \label{lem:DeltaTheta:item1} There exists
  a constant $C$ such that almost-surely, for any $n \geq 0$,
\[
\| \theta_{n+1} - \theta_n \| \leq C \gamma_{n+1} \left(1 + \| S_{n+1}^\sa - \barS(\theta_n) \| \right).
\]
\item \label{lem:DeltaTheta:item2}  There exists a constant $C'$ such that almost-surely, for any $n \geq 0$,
\begin{multline*}
\| \gamma_{n+1} \Psi(\theta_n) - \gamma_n \Psi(\theta_{n-1}) \| \\
 \leq C' \left( \left| \gamma_{n+1} - \gamma_n
  \right| + \gamma_n^2 (1+ \|S_{n}^\sa - \barS(\theta_{n-1}) \|) \right).
\end{multline*}
\item \label{lem:DeltaTheta:item3}  There exists a constant $C"$ such that almost-surely, for any $n \geq 0$,
\begin{multline*}
  \| \gamma_{n+1} T_{\gamma_{n+1},g}\left( \theta_n \right) - \gamma_n
  T_{\gamma_{n},g}\left( \theta_{n-1} \right) \|  \\
  \leq C'' \left( \left| \gamma_{n+1} - \gamma_n \right| + \gamma_n
    \gamma_{n+1}   \right. \\
   \left. + \gamma_n^2 \left(1 + \|S_{n}^\sa - \barS(\theta_{n-1})\| \right)
  \right).
\end{multline*}
  \end{enumerate}
\end{lemma}
\begin{proof}
  The proof of \textit{(\ref{lem:DeltaTheta:item1})} is on the same
  lines as the proof of~\cite[Lemma 15]{atchade2014stochastic}, and is
  omitted.  For \textit{(\ref{lem:DeltaTheta:item2})},  we write by using H\ref{A3}\ref{A3:item2}) and H\ref{A3}\ref{A3:item3}),
  \begin{align*}
&    \| \gamma_{n+1} \Psi(\theta_n) - \gamma_n \Psi(\theta_{n-1}) \| \\
 &  \quad \leq | \gamma_{n+1} - \gamma_n | \, \| \Psi(\theta_n) \| + \gamma_n \| \Psi(\theta_n)-\Psi(\theta_{n-1})\| \\
                                                                    &  \quad \leq | \gamma_{n+1} - \gamma_n | \,  \sup_\Theta \| \Psi \| +  L \, \gamma_n \| \theta_n - \theta_{n-1} \|.
  \end{align*}
We then conclude by \textit{(\ref{lem:DeltaTheta:item1})}. The LHS in \textit{(\ref{lem:DeltaTheta:item3})} is upper bounded by
\begin{multline*}
  \left| \gamma_{n+1} - \gamma_n \right| \sup_{\gamma \in
    \ocint{0,1/L}}
  \sup_{\theta \in \Theta} \| T_{\gamma,g}\left(\theta \right) \|  \\
  + \gamma_n \|T_{\gamma_{n+1},g}\left(\theta_n \right) -
  T_{\gamma_n,g}\left(\theta_{n-1}\right) \|.
\end{multline*}
Under H\ref{A3}, there exists a constant $C$ such that for all
$\gamma, \gamma' \in \ocint{0,1/L}$ and $\theta, \theta' \in \Theta$
(see \cite[Proposition 12]{atchade2014stochastic})
\begin{align*}
& \sup_{\gamma \in \ocint{0,1/L}} \sup_{\theta \in \Theta} \|
T_{\gamma,g}\left(\theta\right) \| < \infty,  \\
& \|
T_{\gamma,g}\left(\theta \right) - T_{\gamma',g}\left(\theta' \right)
\| \leq C \left(\gamma + \gamma' + \|\theta - \theta'\| \right).
\end{align*}
  We then conclude by \textit{(\ref{lem:DeltaTheta:item1})}.
\end{proof}

\begin{lemma}
  \label{lem:PoissonSolution}
  Assume H\ref{A4}. For any $\theta \in \Theta$, there exists a
  function $\hatS_\theta:\Xset \to \rset^q$ such that
  $S - \barS(\theta) = \hatS_\theta - P_\theta \hatS_\theta$ and
  $\sup_{\theta \in \Theta} \sqrtWnorm{ \hatS_\theta } < \infty$.  In
  addition, there exists a constant $C$ such that for any
  $\theta, \theta' \in \Theta$,
\[
\sqrtWnorm{ P_\theta \hatS_\theta - P_{\theta'} \hatS_{\theta'}} \leq
C \, \|\theta - \theta' \|.
\]
\end{lemma}
\begin{proof}
  Set
  $\hatS_\theta(z) \eqdef \sum_{n \geq 0} \left( P_\theta^n S(z) -
    \barS(\theta) \right)$.
  Observe that, when exists, this function satisfies
  $S - \barS(\theta) = \hatS_\theta - P_\theta \hatS_\theta$.  Note
  that under H\ref{A4}\ref{A4:item1})-H\ref{A4}\ref{A4:item2}), there
  exist $C$ and $\rho \in \ooint{0,1}$ such that for any
  $\theta \in \Theta$,
\[
\sum_{n \geq 0} \left\| P_\theta^n S(z) -  \barS(\theta)
\right\| \leq C \, \sqrtWnorm{S} \, \left( \sum_{n \geq 0}
\rho^n \right) \sqrt{W}(z);
\]
the RHS is finite, thus showing that $\hatS_\theta$ exists. This
inequality also proves that
$\sup_\theta \sqrtWnorm{\hatS_\theta}< \infty$. The Lipschitz property
is established in ~\cite[Lemma 4.2.]{fort:moulines:priouret:2011} and
its proof uses H\ref{A4}\ref{A4:item3}).
\end{proof}

\section{Proof of Proposition~\ref{prop:L2norm:DeltaS}}
\label{sec:proof:prop:Lnorm}
Throughout this section, set
$\|U\|_{L_2} \eqdef \PE\left[\|U\|^2 \right]^{1/2}$. By
Lemma~\ref{lem:IterateApproxSto},
$\| S_{n}^\sa - \barS(\theta_{n-1}) \|_{L_2} \leq \sum_{i=1}^3
\mathcal{T}_{i,n}$ with
\begin{align*}
  & \mathcal{T}_{1,n} \eqdef  \Delta_{2:n} \| S_1^\sa - \barS(\theta_0) \|_{L_2},   \\
  & \mathcal{T}_{2,n} \eqdef  \sum_{j=2}^n \Delta_{j:n} \| \barS(\theta_{j-1}) - \barS(\theta_{j-2}) \|_{L_2}, \\
  & \mathcal{T}_{3,n} \eqdef \left\| \sum_{j=2}^n \Delta_{j+1:n} \delta_j \left(m_j^{-1} \sum_{k=1}^{m_j} S(\hid_{k,j-1}) - \barS(\theta_{j-1}) \right) \right\|_{L_2}.
\end{align*}
Since $\Delta_{2:n} \leq \exp(-\delta_\star \sum_{j=2}^n j^{-\beta})$,
then
$$\mathcal{T}_{1,n} = O\left(\exp(-\delta_\star(1-\beta)^{-1} n^{1 -
    \beta}) \right).$$
By H\ref{A3}\ref{A3:item3}), Lemma~\ref{lem:IteratesKernels} and
Lemma~\ref{lem:DeltaTheta}, there exists a constant $C$ such that
$\mathcal{T}_{2,n} \leq C \, \sum_{j=2}^n \Delta_{j:n} \gamma_{j-1}$.
By Lemma~\ref{lem:Vitesse:Sn}, this yields
$\mathcal{T}_{2,n} = O(n^{\beta- \alpha})$. For the last term, we use
a martingale decomposition.

By Lemma~\ref{lem:PoissonSolution}, there exists a function
$\hatS_\theta$ such that
\[
S(\hid_{k,j-1}) - \barS(\theta_{j-1}) = \hatS_{\theta_{j-1}}(\hid_{k,j-1}) -
P_{\theta_{j-1}} \hatS_{\theta_{j-1}}(\hid_{k,j-1}),
\]
and $\sup_{\theta \in \Theta} \sqrtWnorm{ \hatS_\theta }< \infty$.
Hence, we write
\[
m_j^{-1} \sum_{k=1}^{m_j} S(\hid_{k,j-1}) - \barS(\theta_{j-1})
= \partial M_j + R_{j,1} + R_{j,2}
\]
with
\begin{align*}
  & \partial M_j \eqdef m_j^{-1} \sum_{k=1}^{m_j}  \left\{ \hatS_{\theta_{j-1}}(\hid_{k,j-1}) - P_{\theta_{j-1}} \hatS_{\theta_{j-1}}(\hid_{k-1,j-1})  \right\},  \\
  & R_{j,1} \eqdef  m_j^{-1} \left\{  P_{\theta_{j-1}} \hatS_{\theta_{j-1}}(\hid_{0,j-1})  -  P_{\theta_{j}} \hatS_{\theta_{j}}(\hid_{0,j})  \right\},  \\
  & R_{j,2} \eqdef m_j^{-1} \left\{ P_{\theta_{j}} \hatS_{\theta_{j}}(\hid_{0,j})  - P_{\theta_{j-1}} \hatS_{\theta_{j-1}}(\hid_{0,j})  \right\};
\end{align*}
we used that $\hid_{0,j} = \hid_{m_j,j-1}$. Upon noting that
$\partial M_j$ is a martingale-increment, and
$$\hatS_{\theta_{j-1}}(\hid_{k,j-1}) - P_{\theta_{j-1}}
\hatS_{\theta_{j-1}}(\hid_{k-1,j-1}) $$
is a martingale-increment, we have by two successive applications of
\cite[Theorem 2.10]{hall:heyde:1980}: 
\[
\left\| \sum_{j=2}^n \Delta_{j+1:n} \delta_j \partial M_j
\right\|_{L_2} \leq C \left(\sum_{j=2}^n \Delta_{j+1:n}^2
  \frac{\delta_j^2}{m_j} \right)^{1/2}.
\]
By Lemma~\ref{lem:Vitesse:Sn}, this term is $O(n^{-(\beta+c)/2})$. For the second term, we write
\begin{multline*}
  \sum_{j=2}^n \Delta_{j+1:n} \delta_j R_{j,1}  \\
  = \Delta_{3:n} \frac{\delta_2}{m_2} P_{\theta_{1}}
  \hatS_{\theta_{1}}(\hid_{0,1}) - \frac{\delta_n}{m_n} P_{\theta_{n}}
  \hatS_{\theta_{n}}(\hid_{0,n})  \\
  + \sum_{j=2}^{n-1} \left( \Delta_{j+2:n}
    \frac{\delta_{j+1}}{m_{j+1}} - \Delta_{j+1:n} \frac{\delta_j}{m_j}
  \right) P_{\theta_{j}} \hatS_{\theta_{j}}(\hid_{0,j}).
\end{multline*}
By Lemma~\ref{lem:IteratesKernels} and
Lemma~\ref{lem:PoissonSolution}, the RHS is
$O(n^{-(\beta+c)} + n^{-(1+c)})$ so that this second term is
$O(n^{-(\beta+c)})$. Finally, for the third term, by using
Lemma~\ref{lem:IteratesKernels}, Lemma~\ref{lem:DeltaTheta} and
Lemma~\ref{lem:PoissonSolution}, we write
\[
\left\| \sum_{j=2}^n \Delta_{j+1:n} \delta_j R_{j,2} \right\|_{L_2}
\leq \sum_{j=2}^n \Delta_{j+1:n} \frac{\delta_j}{m_j} \gamma_j.
\]
Again by Lemma~\ref{lem:Vitesse:Sn}, this last term is
$O(n^{-(\alpha+c)})$.  Therefore, $\mathcal{T}_{3,n} = O(n^{-(\beta+c)/2})$.

\section{Proof of Theorem~\ref{theo:Cvg:SAEMprox}}
\label{sec:proof:SAEMprox}
Throughout the proof, we will write $S_{n+1}$ instead of
$S_{n+1}^\sa$. 

\paragraph{Proof of Theorem~\ref{theo:Cvg:SAEMprox}}
We prove the almost-sure convergence of the three random sums given in
Theorem~\ref{theo:Cvg:perturbedEMprox}.  The third one is finite
almost-surely since its expectation is finite (see
Proposition~\ref{prop:SumSquare}). The first two ones are of the
form $\sum_n \A_{n+1} \left(S_{n+1} - \barS(\theta_n) \right)$ where $\A_{n+1}$ is respectively
\[
\A_{n+1} = \gamma_{n+1}  \left( T_{\gamma_{n+1},g}\left( \theta_n \right)  \right)', \qquad
\A_{n+1} = \gamma_{n+1} \Psi(\theta_n).
\]
Note that $\A_{n+1} \in \F_n$ (the filtration is defined by
Eq.~\eqref{eq:filtration}). By Lemma~\ref{lem:DeltaTheta} and
H\ref{A3}\ref{A3:item2}-\ref{A3:item3}), for both cases, there exists
a constant $C$ such that almost-surely, for any $n \geq 0$,
\begin{align*}
  & \| \A_{n+1} - \A_n \| \leq C \left( \left|\gamma_{n+1} - \gamma_n
    \right| + \gamma_{n}^2 + \gamma_n \gamma_{n+1}\right) \ \cdots \\
  & \qquad \qquad \times \left(1+
    \|S_n - \barS(\theta_{n-1})\| \right),  \\
  &\|\A_{n+1} \| \leq C \gamma_{n+1}.
\end{align*}
We then conclude by Proposition~\ref{prop:MartingaleReste}.
\begin{proposition} \label{prop:SumSquare} Assume
  H\ref{A4}\ref{A4:item1}) and
  $$\sup_{\theta \in \Theta} \left( \| \Psi(\theta)\| + \| \barS(\theta)\| \right) < \infty.$$
  Then there exists a constant $C$ such that
\[
\sum_n \gamma_{n+1}^2 \PE\left[ \left\| \Psi(\theta_n) \left(
    S_{n+1} - \barS(\theta_n) \right) \right\|^2 \right] \leq C \, \sum_n \gamma_{n+1}^2.
\]
\end{proposition}
\begin{proof}
  We write
\begin{multline*}
  \PE\left[ \left\| \Psi(\theta_n) \left( S_{n+1} -
        \barS(\theta_n) \right) \right\|^2 \right]  \\
  \leq 2 \sup_{\Theta} \| \Psi(\theta) \|^2 \left( \sup_n \PE\left[
      \|S_{n} \|^2 \right]+ \sup_{\Theta} \| \barS \|^2 \right),
\end{multline*}
and conclude by Lemma~\ref{lem:IteratesKernels}.
\end{proof}

\begin{proposition}
  \label{prop:MartingaleReste} Let $\{\theta_n, n\geq 0\}$ be given by
  Algorithm~\ref{algo:SAPG}.  Assume H\ref{A1}, H\ref{A3},
  H\ref{A4}\ref{A4:item1}-\ref{A4:item2}) and
  H\ref{A5}\ref{A5:item1}). In the biased case, assume also
  H\ref{A4}\ref{A4:item3}) and H\ref{A5}\ref{A5:item2}).  Let
  $\{\A_n, n \geq 0 \}$ be a sequence of $d' \times q$ random matrices
  such that for any $n \geq 0$, $\A_{n+1} \in \F_n$, and there exists
  a constant $C_\star$ such that almost-surely
\begin{eqnarray}
&&  \|\A_{n+1} \| \leq C_\star \gamma_{n+1},  \label{eq:HypSurMatriceA:bound1}\\
&& \|\A_{n+1} - \A_n \| \leq C_\star  a_{n+1} \left(1 + \|S_n - \barS(\theta_{n-1})\| \right); \label{eq:HypSurMatriceA:bound2}
\end{eqnarray}
here
$a_{n+1} \eqdef \gamma_n \gamma_{n+1} + \gamma_n^2 + | \gamma_{n+1} -
\gamma_n |$.
Then, almost-surely, the series
$\sum_n \A_{n+1} \left( S_{n+1} - \barS(\theta_n) \right)$ converges.
\end{proposition}
By Lemma~\ref{lem:IterateApproxSto} applied with $\sigma_n = \barS(\theta_n)$, we
decompose this sum into four terms:
\begin{align*}
  \mathcal{T}_{1} & \eqdef \sum_{n \geq 2}  \Delta_{2:n} \A_{1} \left( S_1 - \barS(\theta_0) \right)  \\
                  & = \D_2 \A_{1} \left( S_1 - \barS(\theta_0) \right), \\
  \mathcal{T}_{2} & \eqdef \sum_{n \geq 2} \sum_{j=2}^{n} \Delta_{j:n} \left(\A_j - \A_{j-1} \right) \left(S_{j-1} - \barS(\theta_{j-2}) \right)  \\
                  & = \sum_{j \geq 2} \D_j  \left(\A_j - \A_{j-1} \right) \left(S_{j-1} - \barS(\theta_{j-2}) \right), \\
  \mathcal{T}_3 & \eqdef \sum_{n \geq 2} \sum_{j=2}^{n} \Delta_{j:n} \A_j \left( \barS(\theta_{j-2}) - \barS(\theta_{j-1}) \right) \\
                  &  = \sum_{j \geq 2}  \D_j \A_j \left( \barS(\theta_{j-2}) - \barS(\theta_{j-1}) \right), \\
  \mathcal{T}_4 & \eqdef \sum_{n \geq 2} \sum_{j=2}^{n} \delta_j  \Delta_{j+1:n} \A_j \left( m_j^{-1} \sum_{k=1}^{m_j} S(\hid_{k,j-1})- \barS(\theta_{j-1}) \right) \\
                  & = \sum_{j \geq 2} \delta_j (1+\D_{j+1}) \A_j \left(m_j^{-1} \sum_{k=1}^{m_j} S(\hid_{k,j-1})- \barS(\theta_{j-1}) \right).
\end{align*}
We have by using Eq.~\eqref{eq:HypSurMatriceA:bound1},
\[
\D_2 \| \A_{1} \left( S_1 - \barS(\theta_0) \right) \| \leq \D_2 C_\star
\gamma_1 \left(  \|S_1 \| + \sup_\Theta \| \barS \| \right).
\]
By H\ref{A5}\ref{A5:item1}), $\D_2 < \infty$ so the RHS is
finite thus implying that $\mathcal{T}_1$ is finite almost-surely.

Using Eq.\eqref{eq:HypSurMatriceA:bound2}, there exists a constant $C$
such that
\begin{multline*}
  \PE\left[ \sum_{j \geq 2 }  \D_j \| \left(\A_j - \A_{j-1} \right) \left(S_{j-1} - \barS(\theta_{j-2}) \right) \| \right]  \\
  \leq C \left(1 + \sup_n \PE\left[ \|S_n \|^2\right] + \sup_\Theta
    \|\barS \| \right) \sum_{j \geq 2} a_j \D_j.
\end{multline*}
By H\ref{A3}\ref{A3:item2})-H\ref{A3}\ref{A3:item3}),
H\ref{A5}\ref{A5:item1}) and Lemma~\ref{lem:IteratesKernels}, the RHS
is finite thus implying that $\mathcal{T}_2$ is finite almost-surely.

Similarly, there exists a constant $C$ such that
\begin{multline*}
  \PE\left[ \sum_{j \geq 2} \D_j \| \A_j \left( \barS(\theta_{j-2}) -
      \barS(\theta_{j-1}) \right) \| \right]  \\
  \leq C \, \sum_{j \geq 2} \gamma_j \D_j \PE\left[ \|
    \barS(\theta_{j-1}) - \barS(\theta_{j-2}) \|\right].
\end{multline*}
By H\ref{A3}\ref{A3:item3}), the RHS is bounded (up to a
multiplicative constant) by
$\sum_{j \geq 2} \gamma_j \D_j \PE\left[ \| \theta_{j-1} -
  \theta_{j-2} \|\right]$;
and by H\ref{A5}\ref{A5:item1}) and Lemmas~\ref{lem:IteratesKernels} and
\ref{lem:DeltaTheta}, this sum is finite. Hence $\mathcal{T}_3$ is
finite almost-surely.

We give the proof of the convergence of the last term in the biased
case:
$\PE\left[S(\hid_{k,n}) \vert \F_n \right] \neq \barS(\theta_n)$.  The
proof in the unbiased case corresponds to the following lines with
$R_{j,1} = R_{j,2} =0$ and $\hatS_\theta = S$.  Set
$\overline{\D}_j \eqdef \delta_j (1+\D_{j+1})$.  By
Lemma~\ref{lem:PoissonSolution}, there exists $\hatS_\theta$ such that
\begin{multline*}
  S(\hid_{k,j-1}) - \barS(\theta_{j-1}) \\
  = \hatS_{\theta_{j-1}}(\hid_{k,j-1}) - P_{\theta_{j-1}}
  \hatS_{\theta_{j-1}}(\hid_{k,j-1}),
\end{multline*}
and $\sup_{\theta \in \Theta} \sqrtWnorm{ \hatS_\theta }< \infty$.
Hence, we have
$$\mathcal{T}_4 = \sum_{j \geq 2} \overline{\D}_j \A_j \left( \partial M_j +
  R_{j,1} + R_{j,2} \right),$$ where
\begin{align*}
  & \partial M_j \eqdef m_j^{-1} \sum_{k=1}^{m_j} \left( \hatS_{\theta_{j-1}}(\hid_{k,j-1}) -
    P_{\theta_{j-1}} \hatS_{\theta_{j-1}}(\hid_{k-1,j-1}) \right), \\
  &  R_{j,1} \eqdef m_j^{-1} \left( P_{\theta_{j-1}} \hatS_{\theta_{j-1}}(\hid_{0,j-1}) - P_{\theta_{j}} \hatS_{\theta_{j}}(\hid_{0,j}) \right),  \\
  & R_{j,2} \eqdef m_j^{-1} \left( P_{\theta_{j}} \hatS_{\theta_{j}}(\hid_{0,j}) - P_{\theta_{j-1}} \hatS_{\theta_{j-1}}(\hid_{0,j}) \right).
\end{align*}
Upon noting that
$\PE\left[ \A_j \partial M_j \vert \F_{j-1} \right] = 0$, the
almost-sure convergence of the series
$\sum_j \overline{\D}_j \A_j \partial M_j$ is proved by checking
criteria for the almost-sure convergence of a martingale. By
\eqref{eq:HypSurMatriceA:bound1}, there exists a constant $C$ such that
\begin{multline*}
\sum_j \overline{\D}_j^2  \PE\left[ \| \A_j  \partial M_j \|^2  \right] \leq C \, \sum_j \frac{\gamma_j^2}{m_j^2} \overline{\D}_j^2 \cdots \\
\quad \times   \PE\left[ \left\| \sum_{k=1}^{m_j} \left( \hatS_{\theta_{j-1}}(\hid_{k,j-1}) -
    P_{\theta_{j-1}} \hatS_{\theta_{j-1}}(\hid_{k-1,j-1})\right)  \right\| ^2 \right].
\end{multline*}
By H\ref{A5}\ref{A5:item1}), Lemma~\ref{lem:IteratesKernels} and
\cite[Theorem 2.10]{hall:heyde:1980}, the RHS is finite. \cite[Theorem
2.17]{hall:heyde:1980} implies that
$\sum_j \overline{\D}_j \A_j \, \partial M_j$ is finite
almost-surely. For the second term, we write
\begin{multline*}
  \sum_{j \geq 2} \overline{\D}_j \A_j R_{j,1} = m_2^{-1}
  \overline{\D}_2 \A_2
  P_{\theta_1} \hatS_{\theta_1}(\hid_{0,1})  \\
  + \sum_{j \geq 2} \left( m_{j+1}^{-1} \overline{\D}_{j+1} \A_{j+1} -
    m_j^{-1} \overline{\D}_j \A_j \right) P_{\theta_{j}}
  \hatS_{\theta_{j}}(\hid_{0,j}),
\end{multline*}
so that, by Lemmas~\ref{lem:IteratesKernels} and
\ref{lem:PoissonSolution}, this series is finite almost-surely if
$\sum_j \PE\left[ \| m_{j+1}^{-1} \overline{\D}_{j+1} \A_{j+1} -
  m_j^{-1} \overline{\D}_j \A_j \| \right] < \infty$.
From Eq.~\eqref{eq:HypSurMatriceA:bound1} and Eq.~\eqref{eq:HypSurMatriceA:bound2}, there exists a constant $C$ such that
\begin{multline*}
  \| m_{j+1}^{-1} \overline{\D}_{j+1} \A_{j+1} - m_j^{-1} \overline{\D}_j \A_j \|  \\
  \leq C \gamma_{j+1} \left|m_{j+1}^{-1} \overline{\D}_{j+1} -
    m_j^{-1} \overline{\D}_j \right|  \\
  + m_j^{-1} \overline{\D}_j a_{j+1} \left(1 + \|S_j
    -\barS(\theta_{j-1}) \| \right).
\end{multline*}
H\ref{A5} and Lemma~\ref{lem:IteratesKernels} imply that
$$\sum_j \PE\left[ \| m_{j+1}^{-1} \overline{\D}_{j+1} \A_{j+1} -
  m_j^{-1} \overline{\D}_j \A_j \| \right] < \infty.$$

Finally, by \eqref{eq:HypSurMatriceA:bound1}, Lemmas~\ref{lem:IteratesKernels} to \ref{lem:PoissonSolution}, there exists a constant $C$ such that
\[
\sum_{j \geq 2} \PE\left[  \overline{\D}_j \, \| \A_j R_{j,2} \| \right]  \leq  C\, \sum_{j \geq 2} \gamma_j^2 m_j^{-1} \overline{\D}_j.
\]
The RHS is finite by H\ref{A5} thus implying that
$\sum_j \overline{\D}_j \A_j R_{j,2}$ is finite almost-surely.

\bibliographystyle{plainnat}
\bibliography{BIB}

\end{document}